\documentclass[12pt,a4paper]{article}
\usepackage[T1]{fontenc}
\usepackage[utf8]{inputenc}
\usepackage{amsmath, amssymb, amsthm}
\usepackage{natbib}
\usepackage{booktabs}
\usepackage{threeparttable}
\usepackage{graphicx}
\usepackage[colorlinks=true,linkcolor=blue,citecolor=blue]{hyperref}
\usepackage{titlesec}
\usepackage{lmodern}
\usepackage{ragged2e}
\usepackage{rotating}
\usepackage{adjustbox}
\newtheorem{proposition}{Proposition}

\titleformat{\section}{\large\bfseries}{\thesection}{1em}{}
\titleformat{\subsection}{\normalsize\bfseries}{\thesubsection}{1em}{}

\title{\vspace{-2.5 cm}
\Large \textbf{Reputation, Exposure, and Exit: \\ Organizational Turnover after \#MeToo}
}

\author{
Roy Baharad\thanks{Hebrew University of Jerusalem}
\and
Asaf Eckstein\thanks{Hebrew University of Jerusalem}
\and
Gideon Parchomovsky\thanks{Hebrew University of Jerusalem and University of Pennsylvania}
\and
Rok Spruk\thanks{University of Ljubljana \& University of Western Australia}
}
\bigskip

\date{\vspace{0.6 cm}\today}

\begin{document}

\maketitle

\begin{abstract}
\noindent
We study how economy-wide reputational shocks reshape corporate governance by analyzing board resignations following the \#MeToo movement. We conceptualize the October 2017 revelations surrounding Harvey Weinstein as a common information shock that raised the expected cost of misconduct and intensified scrutiny across firms. Identification exploits cross-sectional heterogeneity in pre-shock exposure, measured by the frequency of Item 5.02 Form 8-K filings, which proxy for firms' sensitivity to governance-related disclosure and reputational risk.  We develop a model of board exit in which directors respond to shifts in reputational pressure through dynamic, belief-driven hazard rates, generating heterogeneous and potentially nonlinear responses. Empirically, we implement a continuous-treatment difference-in-differences design, complemented by dynamic event-study analysis and dynamic matrix completion estimates.  We find that firms with greater exposure experience significantly larger increases in resignation flows following the shock, with effects concentrated in the immediate aftermath and amplified through board-level interactions. The results provide causal evidence that reputational shocks can induce rapid and systematic governance turnover, highlighting the role of information and exposure in shaping institutional responses.
\end{abstract}


\titlespacing*{\section}{0pt}{2.5ex plus 1ex minus .2ex}{1.5ex}
\titlespacing*{\subsection}{0pt}{2ex}{1ex}

\section{Introduction}

The \#MeToo Movement transformed the governance environment in which public corporations operate. Beginning with the October 2017 revelations surrounding Harvey Weinstein, the resulting cascade of allegations did more than exposing particular instances of sexual misconduct: they changed the social and institutional susceptibility of organizations. The movement thus altered the perceived consequences of organizational inaction or misbehavior at large, mainly in raising their reputational cost.
In this paper, we study the corporate-governance implications of that shock by estimating the effect of \#MeToo on executive departure. It should be stressed at the outset that the central claim is not that every post-\#MeToo departure reflects the revelation of sexual misconduct. This argument cannot be made because under securities law, firms report director and officer departures through Item 5.02 of Form 8-K, but they are not generally required to disclose the reasons for those departures \citep{eckstein2025silence}. Decisions such as \textit{Lopez v. CTPartners} which declined to recognize an affirmative securities-law duty to disclose sexual harassment allegations, reinforce this point. 
\\
\\
This paper, therefore, makes no attempt to attribute departures to misconduct. Such an exercise would require information that is not systematically available. Instead, we are focused on a broader point and inquire whether \#MeToo produce corporate-governance spillovers? Put differently, we ask if the movement's transformation of reputational expectations and accountability norms alter the departure behavior in public companies, especially firms that were already more exposed to governance-sensitive disclosure environments. The question of interest, then, is whether firms more exposed to the governance and reputational channels through which \#MeToo operated experienced systematically larger changes in exit behavior after the movement changed the corporate environment.
\\
\\
Governance spillovers are central here, since social movements often affect institutions beyond the immediate domain in which they concentrated. \#MeToo was nominally a movement against sexual harassment and misconduct, but once a new norm of accountability becomes more salient, it could permeate general corporate practice. Boards and managers may become more attentive not only to sexual misconduct, but also to corporate culture, ethical lapses, internal complaints and other reputational risks. All are practices that, prior to the social change creditable to \#MeToo, might have been tolerated or quietly managed. Due to such social change, however, investors and employees may update their expectations about what qualifies as a serious governance failure. In this way, we treat \#MeToo not merely as a misconduct revelation shock, but as a broader accountability shock.
\\
\\
This distinction is important both theoretically and empirically. If \#MeToo only revealed previously hidden misconduct, then one might expect its governance effects to be concentrated in firms where such misconduct was actually uncovered. On the other hand, if \#MeToo also permeated the broader norms of corporate accountability, its effects should extend more broadly: a board may accept a resignation or remove an official more quickly, internal reviews might be initiated more aggressively, and overall, the corporation may become less tolerant of any reputationally costly leadership issue. The immediate empirical implication, then, is that the effect of \#MeToo may be visible in aggregate departure patterns even when the public record does reveal the reason for each departure.
We treat the Weinstein revelations as a common reputational and informational shock to public corporations. 
\\
\\
The shock was common because it affected the entire corporate environment rather than a discrete set of treated firms. But crucially, firms were not equally exposed to that shock. Some firms were more embedded in governance environments in which leadership turnover, formal disclosure, and reputational scrutiny were already more salient before \#MeToo. We hypothesize that these firms were likely to be more sensitive to a shift in accountability norms. We capture this heterogeneity using firms' pre-shock frequency of Item 5.02 filings. Our theoretical model and empirical strategy follow directly from this conceptual framework. Because \#MeToo was an economy-wide shock, there is no natural untreated group. The relevant comparison is therefore not between firms affected by the movement and firms unaffected by it. Rather, it is between firms with greater and lesser pre-shock exposure to governance-related disclosure. We ask whether high-exposure firms experienced larger changes in resignation and termination flows after October 2017, compared to low-exposure firms. The design thus considers exposure as a continuous treatment intensity and estimates whether the post-shock change in departures varies with that pre-shock exposure.
\\
\\
In terms of interpretation, we should again say that a positive effect should not be read as evidence that high-exposure firms had more undisclosed cases of sexual misconduct. This would be too strong and, given the available data, unsupported. The more appropriate interpretation is actually more propitious: firms more exposed to governance and reputational accountability responded more strongly to the post-\#MeToo environment. The mechanism points to a change in governance norms: increased scrutiny, heightened sensitivity to reputational risks and stronger shareholder expectations.
\\
\\
We implement our empirical strategy using a continuous-treatment difference-in-difference design, which is then supplemented by dynamic event-study analysis and dynamic matrix completion methods. The outcome of interest is naturally the change in departure intensity rather than the mere level of departures, as firms differ in their ordinary turnover regimes: some routinely experience more transitions than others for reasons orthogonal to \#MeToo. By focusing on changes in resignation and termination flows, the analysis aims to isolate whether the discontinuous shift in exit behavior produced by the post October 2017 environment was stronger among more exposed firms. 
\\
\\
The paper makes two principal contributions. First, it contributes to the literature on social movements and corporate governance, showing how a movement that did not in itself change corporate law nonetheless altered corporate behavior. Second, the paper contributes methodologically to the study of common shocks, using exposure-based design for studying how common shocks generate heterogenous institutional response. 
Structurally, the paper proceeds as follows. Part 2 provides the background on the \#MeToo movement, reviews the related literature and situates the paper's contributions within it. Part 3 develops a theoretical model linking reputational pressure, belief updating, norm diffusion and board exit. Part 4 presents the empirical strategy. Part 5 describes the data and collection process. Part 6 presents the results. Part 7 discusses the implications. Part 8 concludes.

\section{Background and Related Literature}

The term ``Me Too'' originated with activist Tarana Burke, who is also credited with being the founder of the movement on account of her years-long campaign against sexual abuse and work with survivors \citep{garcia2017burke}. However, the turning point in the life of the movement came over a decade later. On October 5, 2017, journalists Jodi Kantor and Megan Towhey of The New York Times published a story in which they accused producer Harvey Weinstein of decades of predatory behavior \citep{kantor2017weinstein}. In that story, actress Ashley Judd shared information about unwanted sexual advancements made at her by Weinstein. A few days later, Journalist Ronan Farrow of The New Yorker followed up with his own story on Weinstein, offering additional accounts of his involvement in sexual abuse \citep{farrow2017weinstein}. Reacting to these stories, actress Alyssa Milano sent out a tweet, asking women to tweet the hashtag \#MeToo if they experienced sexual harassment or assault. Milano's tweet changed the world forever. Within hours, Milano received tens of thousands of tweets from individuals who responded to her call. Within 24 hours, Facebook counted over 12 million messages and communications sparked by Milano's call \citep{dorking2017metoo}.
\\
\\
The movement not only gave voice to multiple women who suffered sexual abuse, but also exposed the scope and pervasiveness of the problem. The movement's impact quickly transcended the U.S., uniting women all over the world. Why was the \#MeToo movement so impactful? Scholars attribute its success to the fact that \#MeToo provided a much-needed informal complaint channel that served as an alternative to institutionalized complaint channels, such as the police and courts. The use of the hashtag ``MeToo'' enabled sexual abuse victims to utilize social platforms to accuse and shame aggressors. \cite{tuerkheimer2019beyond} argued that one of the central advantages of informal channels over the formal ones consisted of the ability to empower the victim by giving her the agency to enter into a new system over which she has control. Victims were given the option to enter into survivor-friendly spaces, in which there isn't free access to third parties, or into one such as \#MeToo, which Tuerkheimer labels as ``The New Court of Public Opinion,'' in which there is open access and stronger retaliation against the aggressor. 
\\
\\
The victims' stories, shared via social media, led the media to intensify its focus on sexual misconduct by other powerful men. An additional powerful person was accused of abusive behavior every 20 minutes in the following weeks \citep{kannan2017powerful}. The \#MeToo movement has spanned many industries including fashion, music, sports, entertainment, architecture, philanthropy, academia, and politics \citep{tuerkheimer2019beyond}. But its impact on the workplace has been even greater. Many of the transgressions involved in the \#MeToo movement happened in the workplace. The victims were either subordinates of the aggressors or employed by the same company as the offenders, and many of the \#MeToo cases resulted in substantial lawsuits that listed the individual aggressors and the companies for which they worked as joint defendants \citep{fotiadis2024employerstoo}. Moreover, companies had to pay out significant amounts of money to women-plaintiffs in order to settle those cases, which resulted in many liability insurance carriers leaving the market, and with others stipulating stricter coverage exclusions and charging enhanced premiums \citep{meyers2021insurance}. \cite{tippett2018legal} pointed out that the legal risk to companies is likely to grow in the future as courts update the standard of what qualifies as ``severe or pervasive'' harassment. In response to the increased exposure to liability, companies began implementing enhanced safeguards and educational programs designed to ensure a safe working environment for women \citep{hymowitz2017workplace}. 
\\
\\
Likewise, studies of the stock market revealed a negative reaction to \#MeToo cases. By way of example, \cite{borellikjaer2021metoo} found \textit{``a negative 1.5\% cumulative abnormal return over the event day and the following trading day, which corresponds to an average impact of \$450 million for the [199] companies in [the] sample.}'' The authors point out that the ``effect is significant and robust across a wide array of inference tests and model choices.'' They further note that ``the involvement of a CEO is consistently a strong driver of negative abnormal returns (a further of 5\% point drop in return on top of the average effect)'' and ``that abnormal increases in media coverage of the firm around the scandal strongly associate with negative return (around -5\% points per standard deviation increase in coverage), whereas instances where companies self-disclose the misconduct are found to be punished less by the market (3\% points less so).'' \cite{bouzzine2022reputation} likewise report that sexual misconduct accusations raised against executives had resulted in substantial decline in stock returns. 
\\
\\
Our approach corresponds to extant academic writings estimating the effects of the \#MeToo movement across multiple domains. \cite{gertsberg2025unintended}, for instance, has demonstrated that the movement resulted in reduced female work-productivity caused by a decrease in collaborative work effort with male employees, explaining that such a decrease stems from men's perception of a higher risk of sexual harassment accusations. Gertsberg frames this outcome as an ``unintended consequence'' of the \#MeToo movement, damaging career opportunities of women. Similar arguments have been raised in business contexts. A more favorable result was reported by \cite{levy2025effects}, who studied the effect of the movement on reported sex crimes, finding that upon the outburst of \#MeToo, reported sex crimes around the world increased by 10\%. According to their study, the \#MeToo movement gave victims the courage to come forth and reveal the abuse they had suffered. 
\\
\\
Previous works likewise paid attention to the effect of \#MeToo on closely related issues in corporate law and governance. As noted, \cite{borellikjaer2021metoo} have investigated the impact of the movement on stock prices. One of their findings indicates that the movement substantially increased the odds for a misconduct scandal to be uncovered. Following \#MeToo, the authors report, the average number of sexual harassment scandals has quadrupled from 1.4 to 5.7 per month. \cite{cici2021mutualfund} suggest that women employees' productivity has improved pursuant to \#MeToo, as the movement reduced the risk of sexual harassment that prevented them from optimally extracting their human capital. \cite{fairchild2023gender} identify that both investors and consumers tend to respond negatively to reactive, symbolic steps taken by corporations following \#MeToo.
\\
\\
Another interesting insight was offered by \cite{billings2022investors}, who noted that post-\#MeToo, investors started perceiving inclusive culture as beneficial, rewarding companies that historically included women in boards, whereas companies with history of exclusion were suffering negative market reactions. Along similar lines, \cite{lins2025sexism} suggest that companies with a culture of inclusion have earned higher stock returns after the outburst of \#MeToo.
\\
\\
Contributors have likewise discussed observable changes in companies even without using financial instruments per se. As mentioned previously, \cite{tippett2018legal} studies the movement's effect on the odds of confronting liability due to changes it induced in employment discrimination law. Furthermore, \cite{miazad2021sex} delves into corporate practices to study how workplace culture has changed pursuant to the movement, including board diversity, pay gaps, denial of employee compensation in case of termination due to sexual harassment, and more. In addition, \cite{meyers2021insurance} explore the involvement of liability insurers in mitigating sexual harassment risks by conditioning coverage on companies' active efforts to eradicate them. \cite{windemuth2019boilerplate} identifies that corporate lawyers responded to risks associated with \#MeToo-induced norm-shifting by adding an explicit clause in mergers and acquisitions agreements, by which the target company declares that it is unaware of any sexual misconduct allegations raised against its senior officials. \cite{boyle2018social} suggest that the movement dramatically changed corporate culture, with 79 percent of firms reporting distinct differences between the pre- and post-\#MeToo era.
\\
\\
Some have also attempted to focus on the effect of \#MeToo on departures, but treatments seem to remain mainly anecdotal. In other words, some have used norm-shifting following \#MeToo as a potential explanation of anomalies. For instance, \cite{challenger2020ceo} identified increased CEO exits in 2019, suggested that \textit{``[f]ollowing the \#MeToo movement, companies were determined to hold CEOs accountable for lapses in judgment pertaining to professional and personal conduct, creating higher standards at the C-level. What may have gone unrecognized or was downplayed in the past was not overlooked by boards, shareholders, or the general public in 2019.''}

\section{Theoretical Mechanism: Reputational Pressure, Beliefs, and Board Exit}

This section develops a microeconomic foundation for board resignation decisions under reputational shocks. The objective is to derive a structural mapping between an economy-wide information shock, the \#MeToo movement, and firm-level resignation flows, and to provide a formal interpretation of the reduced-form estimands. The framework integrates dynamic discrete choice, Bayesian belief updating, and hazard-based exit behavior in a unified setting. A central implication of the model is that exposure operates as a sufficient statistic governing the elasticity of resignation responses to a common shift in the informational and reputational environment, consistent with recent work on heterogeneous treatment effects under aggregate shocks \citep{chernozhukov2021causal}.

\subsection{Environment and Timing}

We begin by specifying the decision problem faced by individual directors and senior executives. Consider a firm $i$ with a board indexed by $j \in \{1,\dots,J_i\}$ and discrete time periods $t = 1,2,\dots,T$. In each period, directors simultaneously decide whether to remain affiliated with the organization or resign. Once a director exits, the decision is irreversible and the individual no longer participates in subsequent periods.
\\
\\
This formulation reflects an important institutional feature of board turnover. Directors do not continuously adjust their level of organizational involvement. Instead, departures typically occur as discrete organizational events triggered by changes in incentives, expectations, or perceived risks. The relevant economic decision is therefore not how much effort a director supplies, but whether continued affiliation with the organization remains privately optimal. The distinction is important because large informational shocks such as the \#MeToo movement may alter incentives without necessarily changing the underlying structure of the organization itself.
\\
\\
The decision problem is inherently forward-looking. Directors derive benefits from remaining affiliated with a firm, including compensation, prestige, influence over strategic decision-making, access to professional networks, and future career opportunities. At the same time, remaining affiliated exposes directors to reputational, legal, and professional risks that may materialize in the future. Consequently, resignation decisions depend not only on current organizational conditions but also on expectations regarding future scrutiny and future information revelation.
\\
\\
To formalize this trade-off, let $d_{ijt} \in \{0,1\}$ denote the director's action, where $d_{ijt}=1$ corresponds to remaining and $d_{ijt}=0$ corresponds to resignation. Let $V_{ijt}$ denote the value of remaining affiliated with firm $i$ at time $t$. Normalizing the value of resignation to zero, the continuation value satisfies the Bellman equation

\begin{equation}
V_{ijt}
=
\max
\left\{
0,\;
u_{ijt}
+
\beta \mathbb{E}_t[V_{ij,t+1}]
\right\},
\end{equation}

where $\beta \in (0,1)$ denotes the discount factor and $u_{ijt}$ represents the current-period utility flow associated with continued affiliation. The first term inside the maximization operator corresponds to resignation, while the second term captures the value of remaining on the board. Directors therefore remain affiliated whenever the expected present value of future benefits exceeds the value of exiting.
\\
\\
Equation (1) highlights a central feature of the model. Directors do not react solely to contemporaneous events. Instead, they continuously reassess the future consequences of remaining affiliated with the organization. This forward-looking structure is essential for understanding the effects of the \#MeToo movement. The shock need not generate immediate firm-specific revelations in order to affect behavior. Rather, the movement may alter expectations regarding future investigations, future information disclosure, stakeholder scrutiny, litigation risk, or reputational sanctions. As a consequence, directors may choose to resign even when current organizational conditions remain unchanged because the expected future cost of affiliation has increased.
\\
\\
This interpretation distinguishes the mechanism developed here from models in which organizational turnover occurs only in response to realized misconduct. In the present framework, departures may arise from changes in expectations rather than changes in fundamentals. The \#MeToo movement therefore operates primarily as a forward-looking informational shock that changes how directors evaluate future organizational risks. This feature is particularly important for understanding why organizational responses may emerge rapidly following a common shock and why effects may persist over multiple periods as beliefs continue to evolve.
\\
\\
The dynamic structure follows the canonical framework of forward-looking discrete choice models developed by \citet{rust1987optimal} and subsequently extended by \citet{aguirregabiria2010dynamic}. The innovation in the present setting is that the continuation value is directly affected by changes in the informational and reputational environment, allowing aggregate social movements to generate heterogeneous organizational responses across firms with different levels of exposure.

\subsection{Payoffs and Reputational Risk}

We next specify how changes in the informational environment enter directors' payoff functions. As discussed above, directors derive benefits from continued affiliation with the firm but also incur costs associated with governance responsibilities, reputational exposure, and potential future sanctions. The \#MeToo movement affects behavior by altering the magnitude and perceived likelihood of these costs. The per-period payoff from remaining affiliated with the organization is given by

\begin{equation}
u_{ijt}
=
B_{ijt}
-
C_{ijt},
\end{equation}

where $B_{ijt}$ denotes the benefits associated with board membership and $C_{ijt}$ denotes the expected costs of continued affiliation. The benefit component encompasses both pecuniary and non-pecuniary rewards. These include direct compensation, prestige associated with board membership, influence over strategic decision-making, access to valuable professional networks, and future career opportunities generated through organizational affiliation. In many settings these benefits are substantial and help explain why directors do not resign in response to ordinary fluctuations in firm performance or routine governance disputes. The central mechanism of the model operates through the cost component. Specifically, expected costs are given by

\begin{equation}
C_{ijt}
=
C_{ijt}^{0}
+
\phi_i \cdot \rho_t \cdot \mathbb{E}_t[\theta_{ijt}]
+
\epsilon_{ijt}.
\end{equation}

This decomposition separates three conceptually distinct sources of cost. The first component, $C_{ijt}^{0}$, captures baseline costs that exist independently of the informational environment. These may include ordinary governance responsibilities, monitoring obligations, legal liabilities associated with board service, opportunity costs of time, and personal circumstances unrelated to reputational concerns. Such costs are present even in the absence of major informational shocks and constitute the normal burden of organizational oversight.
\\
\\
The second component captures reputational exposure and represents the central mechanism of the model. It consists of the interaction between firm-level exposure, aggregate reputational pressure, and directors' beliefs regarding latent organizational vulnerability. The multiplicative structure is intentional. It implies that the consequences of an aggregate reputational shock depend not only on the magnitude of the shock itself but also on the firm's underlying exposure to reputational scrutiny.
\\
\\
The parameter $\phi_i$ captures firm-specific exposure to governance-sensitive reputational risk. Conceptually, $\phi_i$ summarizes organizational characteristics that make a firm particularly vulnerable to adverse information once public scrutiny intensifies. Firms characterized by frequent governance-related disclosures, extensive public visibility, complex organizational structures, or histories of executive turnover are likely to exhibit larger values of $\phi_i$ than otherwise similar firms operating in more stable governance environments.
\\
\\
Importantly, exposure should not be interpreted as misconduct itself. Rather, it captures the extent to which organizational conditions make future reputational costs particularly sensitive to changes in the informational environment. Two firms may exhibit identical underlying governance quality while differing substantially in exposure if one operates in a setting where governance-related information is more likely to become publicly visible or economically consequential.
\\
\\
The variable $\rho_t$ represents the economy-wide reputational price attached to adverse information. A useful interpretation is that $\rho_t$ measures the expected social, legal, regulatory, and economic consequences associated with organizational misconduct once such information becomes publicly salient. When $\rho_t$ is low, governance failures generate relatively limited consequences. When $\rho_t$ is high, identical governance failures become substantially more costly.
\\
\\
The final component, $\mathbb{E}_t[\theta_{ijt}]$, captures directors' beliefs regarding latent reputational vulnerability. This term reflects the perceived probability that information capable of generating reputational, legal, or governance-related costs will become relevant in the future. Importantly, vulnerability itself need not change. Organizational responses may arise because beliefs change even when underlying organizational conditions remain constant. This distinction is central to the interpretation of the \#MeToo movement as an informational shock rather than merely a sequence of firm-specific events.
\\
\\
The specification in equation (3) draws on the literature on reputational incentives, career concerns, and dynamic organizational behavior \citep{holmstrom1999managerial,fudenberg1992maintaining}. In these frameworks, agents internalize expected future penalties associated with adverse information and adjust behavior accordingly. The novelty here is that reputational costs are allowed to depend on both aggregate informational conditions and firm-specific exposure, thereby generating heterogeneous responses to a common shock.
\\
\\
We model the \#MeToo movement as a structural break in the aggregate reputational environment:

\begin{equation}
\rho_t
=
\begin{cases}
\rho_0, & t<T_0,\\
\rho_1, & t\geq T_0,
\end{cases}
\qquad
\rho_1>\rho_0,
\end{equation}

where $T_0$ corresponds to October 2017. This representation captures the idea that the Weinstein revelations fundamentally altered the expected consequences associated with misconduct and organizational inaction. Following the shock, investors, employees, regulators, media organizations, and boards became substantially more attentive to signals of misconduct and more willing to impose sanctions on organizations perceived as failing to address such concerns. The movement therefore increased the reputational price attached to adverse information even for firms in which no new allegations immediately emerged.
\\
\\
An important implication follows immediately. Because firms differ in exposure, identical increases in aggregate reputational pressure generate heterogeneous organizational responses. Differentiating equation (3) with respect to $\rho_t$ yields

\begin{equation}
\frac{\partial C_{ijt}}
{\partial \rho_t}
=
\phi_i
\cdot
\mathbb{E}_t[\theta_{ijt}],
\end{equation}

which is increasing in firm-level exposure. The same increase in reputational pressure therefore generates larger increases in expected costs among highly exposed firms. This simple comparative static provides the fundamental intuition underlying the empirical design.
\\
\\
In the empirical analysis, the latent exposure parameter $\phi_i$ is proxied using Item 5.02 SEC Form 8-K disclosure intensity measured prior to October 2017. Because this measure is constructed entirely from pre-shock disclosures, it captures ex ante organizational exposure rather than post-shock organizational responses. The exposure parameter therefore serves as the critical bridge between the theoretical framework and the empirical identification strategy developed below.

\subsection{Bayesian Updating and Information Structure}

The model developed thus far treats reputational costs as functions of expectations. We now make explicit how those expectations are formed and how they evolve following a large-scale informational shock. A central feature of the framework is that directors operate under incomplete information. Board members rarely possess perfect knowledge regarding all dimensions of organizational behavior. Even within well-governed firms, directors observe organizational conditions only through a collection of imperfect signals generated by internal reporting systems, employee complaints, litigation, regulatory investigations, media coverage, and broader public discourse. As a consequence, directors cannot directly observe the firm's true level of reputational vulnerability. Instead, they must infer this vulnerability from available information.
\\
\\
This informational problem becomes particularly important in the presence of large-scale social movements. Events such as the \#MeToo movement do not merely reveal new information about specific organizations. More fundamentally, they alter the informational environment through which organizational behavior is interpreted. Information that may previously have been ignored, discounted, or viewed as idiosyncratic suddenly becomes more credible, more salient, and more informative regarding future organizational outcomes.
\\
\\
The \#MeToo movement can therefore be interpreted not only as a reputational shock but also as an informational shock. Following the Weinstein revelations, allegations of workplace misconduct were no longer viewed as isolated events. Instead, stakeholders increasingly interpreted such allegations as signals of broader organizational failures in governance, oversight, and accountability. This change altered how directors evaluated future reputational risk even in firms where no new allegations immediately emerged.
\\
\\
To formalize this mechanism, let

\[
\theta_{ijt} \in [0,1]
\]

denote the latent reputational vulnerability associated with director $j$ in firm $i$ at time $t$. The variable $\theta_{ijt}$ should be interpreted broadly as the probability that information capable of generating reputational, legal, or governance-related costs becomes relevant in the future. Importantly, this vulnerability is not directly observable. Directors therefore form beliefs regarding $\theta_{ijt}$ using available information.
\\
\\
We assume that prior beliefs follow a Beta distribution,

\begin{equation}
\theta_{ijt}
\sim
\text{Beta}(a_{it},b_{it}),
\end{equation}

where $a_{it}$ and $b_{it}$ summarize the information accumulated before period $t$. This specification is convenient because it admits a tractable Bayesian updating rule while preserving a natural interpretation of beliefs as probabilities. The prior mean is given by

\[
\frac{a_{it}}
     {a_{it}+b_{it}},
\]

which captures directors' assessment of organizational vulnerability before observing new information. Directors subsequently observe an aggregate public signal, denoted by $s_t$, which summarizes information regarding the prevalence, credibility, and consequences of workplace misconduct. The signal may incorporate media reporting, legal developments, regulatory actions, public allegations, and broader societal discussion. Crucially, the signal is not firm-specific. Rather, it reflects information available to all directors operating within the broader corporate environment.
\\
\\
Let $\kappa_t$ denote the precision attached to this signal. Precision determines the weight directors assign to newly observed information relative to their prior beliefs. Posterior beliefs are therefore given by

\begin{equation}
\mathbb{E}_t[\theta_{ijt}]
=
\frac{
a_{it}
+
\kappa_t s_t
}
{
a_{it}
+
b_{it}
+
\kappa_t
}.
\end{equation}

Equation (7) highlights a central feature of the model. Posterior beliefs represent a weighted average of prior information and newly observed signals. When signal precision is low, directors place relatively little weight on new information and beliefs remain largely anchored by prior assessments. Conversely, when signal precision is high, public information receives greater weight and beliefs become substantially more responsive to newly observed events. The \#MeToo movement is modeled as a discrete increase in informational precision:

\begin{equation}
\kappa_t
=
\begin{cases}
\kappa_0, & t<T_0, \\
\kappa_1, & t\geq T_0,
\end{cases}
\qquad
\kappa_1>\kappa_0.
\end{equation}

The interpretation is straightforward. Prior to October 2017, allegations of workplace misconduct often generated limited informational content regarding broader organizational quality. Following the Weinstein revelations, however, identical allegations became substantially more informative regarding governance failures, oversight deficiencies, and latent organizational vulnerabilities. Consequently, directors rationally placed greater weight on publicly available information when evaluating future reputational risk.
\\
\\
An important implication follows immediately. The \#MeToo shock affects behavior through two distinct channels. The first channel operates through the increase in the reputational price of adverse information, represented by the rise in $\rho_t$ described in the previous subsection. The second channel operates through the increase in informational precision represented by $\kappa_t$. Even if the underlying level of organizational vulnerability remains unchanged, directors may revise beliefs because public information becomes more credible and more informative.
\\
\\
This distinction is critical for interpreting the empirical results. The model does not require new firm-specific allegations to emerge after October 2017. Instead, behavior may change because directors update their beliefs regarding future organizational risks. The same informational shock therefore affects even firms that experience no immediate disclosure event. In this sense, the model provides a microeconomic foundation for why a common societal shock can generate widespread organizational adjustment despite substantial heterogeneity in realized misconduct.
\\
\\
The Bayesian updating mechanism also generates a natural source of persistence. Information arrives gradually, beliefs evolve over time, and organizational responses need not occur immediately. Directors who initially remain affiliated with the firm may subsequently revise their beliefs as additional information accumulates. Consequently, the model predicts that organizational turnover may persist well beyond the initial shock period, a prediction that will play an important role in interpreting the dynamic treatment effects estimated below.

\subsection{Resignation Hazard and Aggregation}

We now connect the directors' optimization problem to observable resignation behavior. Given the payoff structure developed above, resignation emerges as a threshold decision. Directors remain affiliated with the firm as long as the expected benefits of continued affiliation exceed the expected costs. Once the expected cost of remaining becomes sufficiently large relative to the benefits, resignation becomes privately optimal.

Formally, let

\[
R_{ijt}
=
\mathbf{1}\{V_{ijt}=0\}
\]

denote an indicator variable equal to one if director $j$ resigns from firm $i$ in period $t$ and zero otherwise. Resignation therefore occurs whenever the continuation value associated with remaining affiliated falls below the value of exiting.
\\
\\
This formulation highlights an important implication of the model. Resignation need not be triggered by the realization of adverse information itself. Rather, resignation occurs when expectations regarding future costs become sufficiently unfavorable. Consequently, directors may choose to exit even in the absence of new firm-specific allegations if changing beliefs or increased scrutiny alter the expected value of continued affiliation.
\\
\\
To obtain a tractable representation of resignation behavior, assume that the idiosyncratic shock $\epsilon_{ijt}$ follows a standard logistic distribution. Under this assumption, the probability of resignation can be written as

\begin{equation}
\mathbb{P}(R_{ijt}=1)
=
\frac{
1
}{
1+
\exp
\left(
B_{ijt}
-
C_{ijt}^{0}
-
\phi_i\rho_t\mathbb{E}_t[\theta_{ijt}]
\right)
}.
\end{equation}

This representation follows directly from standard random utility models and discrete choice frameworks \citep{mcfadden1974conditional}. The logistic specification is particularly useful because it delivers a smooth mapping between latent reputational pressure and observed resignation behavior. Small changes in beliefs or reputational incentives need not generate immediate exits for all directors. Instead, the probability of resignation increases continuously as expected costs rise relative to expected benefits.
\\
\\
Several comparative statics follow immediately. First, resignation probabilities increase as aggregate reputational pressure $\rho_t$ rises. Second, resignation probabilities increase with firm-level exposure $\phi_i$. Third, resignation probabilities increase when directors become more pessimistic regarding latent organizational vulnerability, as reflected in higher values of $\mathbb{E}_t[\theta_{ijt}]$. These relationships provide the direct behavioral channel linking the \#MeToo shock to organizational turnover.
\\
\\
The logistic structure also highlights an important source of heterogeneity. Directors differ in compensation, career opportunities, reputational capital, and outside options. Consequently, identical informational shocks need not generate identical behavioral responses. Some directors may choose to remain affiliated despite substantial increases in reputational pressure, whereas others may exit following relatively modest changes in expected costs. Aggregate organizational responses therefore emerge from the interaction between common informational shocks and heterogeneous individual incentives.
\\
\\
To connect the model more closely to observed turnover dynamics, define the resignation hazard

\begin{equation}
h_{ijt}
=
\mathbb{P}
\left(
R_{ijt}=1
\mid
R_{ij,t-1}=0
\right).
\end{equation}

The hazard rate measures the conditional probability that a director exits in period $t$ given that the director remained affiliated with the firm through period $t-1$. This object is particularly useful because resignation is fundamentally an event-history process. Directors who have already exited no longer participate in future decisions, implying that turnover is naturally characterized as a sequence of conditional exit probabilities.
\\
\\
The hazard formulation also provides a convenient interpretation of the \#MeToo shock. The movement does not directly determine whether a particular director resigns. Instead, it shifts the hazard of resignation by altering reputational pressure and informational precision. Increases in resignation flows therefore arise because the probability of exit rises among directors who remain at risk of departure. The final step is to aggregate individual resignation hazards into a firm-level outcome. Summing across all directors yields

\begin{equation}
Y_{it}
=
\sum_{j=1}^{J_i}
h_{ijt}.
\end{equation}

The variable $Y_{it}$ represents expected resignation intensity within firm $i$ during period $t$. This aggregation is important because the empirical analysis is conducted at the firm-month level rather than the individual-director level. The model therefore requires a mapping from individual behavioral decisions to observable organizational outcomes.
\\
\\
Several features of the aggregation deserve emphasis. First, resignation intensity is fundamentally different from a binary turnover indicator. A binary indicator records only whether at least one resignation occurred. By contrast, resignation intensity captures the expected number of departures within a given period. This distinction is important because organizational responses to reputational shocks often involve multiple simultaneous departures rather than isolated exits.
\\
\\
Second, the aggregation naturally accommodates heterogeneity in board size. Firms differ substantially in the number of directors and executives potentially exposed to reputational pressure. Summing across individual hazards therefore provides a measure of organizational adjustment that reflects the cumulative response of all individuals operating within the firm's governance structure.
\\
\\
Finally, equation (10) establishes the critical bridge between the theoretical framework and the empirical analysis. The observed outcome variable is not a primitive object but rather the aggregation of individual resignation hazards generated by forward-looking optimization, Bayesian belief updating, and exposure-dependent reputational incentives. Consequently, empirical estimates of changes in resignation intensity can be interpreted as evidence regarding the underlying behavioral mechanisms embedded in the model rather than merely reduced-form correlations between the \#MeToo movement and organizational turnover.

\subsection{Comparative Statics}

Having characterized individual resignation behavior and its aggregation to the firm level, we now examine how changes in the informational environment affect organizational turnover. The comparative statics provide the central theoretical predictions of the model and establish the link between reputational shocks, firm-level exposure, and resignation intensity.
\\
\\
Recall that organizational turnover arises because directors compare the benefits of continued affiliation against expected reputational costs. Since those costs depend on both aggregate reputational pressure and directors' beliefs regarding latent organizational vulnerability, changes in the informational environment alter resignation hazards and, consequently, aggregate resignation intensity.

Let

\[
f_{ijt}
=
\Lambda'
\left(
B_{ijt}
-
C_{ijt}^{0}
-
\phi_i\rho_t\mathbb{E}_t[\theta_{ijt}]
\right)
\]

denote the derivative of the logistic hazard function with respect to its index. Differentiating aggregate resignation intensity with respect to the economy-wide reputational pressure parameter yields

\begin{equation}
\frac{\partial Y_{it}}
{\partial \rho_t}
=
\sum_{j=1}^{J_i}
f_{ijt}
\cdot
\phi_i
\cdot
\mathbb{E}_t[\theta_{ijt}].
\end{equation}

Equation (11) immediately establishes that increases in reputational pressure raise organizational turnover. The effect is larger when directors perceive greater latent vulnerability and when firms exhibit greater exposure to governance-sensitive reputational risk. Intuitively, an increase in reputational pressure raises the expected cost of remaining affiliated with the organization, thereby increasing the likelihood of resignation among directors already concerned about future scrutiny.
\\
\\
More importantly, the model predicts systematic heterogeneity in the response to common shocks. Differentiating once more with respect to firm-level exposure yields

\begin{equation}
\frac{\partial^2 Y_{it}}
{\partial \rho_t \partial \phi_i}
=
\sum_{j=1}^{J_i}
f_{ijt}
\cdot
\mathbb{E}_t[\theta_{ijt}]
+
\sum_{j=1}^{J_i}
f'_{ijt}
\cdot
\phi_i
\cdot
\mathbb{E}_t[\theta_{ijt}]^2.
\end{equation}

Equation (12) is the central comparative static result of the model. The first term captures the direct exposure effect. Firms characterized by greater exposure experience larger increases in resignation intensity following a common increase in reputational pressure. This component alone implies heterogeneous organizational responses to a common shock.
\\
\\
The second term captures nonlinear amplification arising from the curvature of the resignation hazard. Because resignation decisions are threshold-based, directors positioned near the margin of indifference may respond disproportionately to relatively small changes in expected costs. As a consequence, organizational responses need not be linear. In some firms, modest increases in reputational pressure may generate little change in turnover, whereas in others similar increases may trigger substantial waves of resignations.
\\
\\
The nonlinear component is particularly important in the context of the \#MeToo movement. Directors differ substantially in their outside opportunities, career concerns, reputational capital, and expectations regarding future scrutiny. Consequently, the same aggregate informational shock may generate very different organizational responses depending on the distribution of directors relative to their resignation thresholds. This mechanism provides a theoretical rationale for the substantial heterogeneity observed across firms and sectors in the empirical analysis.
\\
\\
An additional implication concerns the role of belief updating. Combining equations (7) and (12) reveals that increases in informational precision magnify the response to reputational shocks by increasing the weight directors assign to adverse public information. Thus, the \#MeToo movement affects resignation behavior not only because misconduct becomes more costly but also because available information becomes more informative. The interaction of these two channels generates a larger organizational response than would arise from either channel in isolation. The preceding discussion can be summarized formally.

\begin{proposition}[Exposure Amplification]
Suppose that directors assign positive probability to latent organizational vulnerability, such that
\[
\mathbb{E}_t[\theta_{ijt}]>0
\]
for all relevant directors. Then an increase in aggregate reputational pressure raises firm-level resignation intensity. Moreover, the magnitude of this effect is increasing in firm-level exposure.
\end{proposition}

\noindent
\textit{Proof.}
Equation (11) implies that
\[
\frac{\partial Y_{it}}
{\partial \rho_t}
>0
\]
whenever
\[
f_{ijt}>0,
\qquad
\phi_i>0,
\qquad
\mathbb{E}_t[\theta_{ijt}]>0.
\]
Equation (12) further shows that the marginal effect of reputational pressure is increasing in exposure through both the direct and nonlinear channels. \hfill $\square$
\\
\\
The proposition provides the fundamental theoretical foundation for the empirical design. The model predicts that a common reputational shock should not affect all firms equally. Instead, firms characterized by greater pre-existing exposure should experience systematically larger increases in resignation intensity. This prediction corresponds directly to the continuous-treatment difference-in-differences specification estimated below, where pre-shock disclosure intensity serves as an observable proxy for the latent exposure parameter $\phi_i$.
\\
\\
The comparative statics also generate an important prediction regarding dynamics. Because belief updating is gradual and organizational responses are nonlinear, turnover need not occur immediately following the shock. Instead, organizational adjustment may unfold over multiple periods as directors revise beliefs and reassess the expected costs of continued affiliation. This prediction provides a direct theoretical rationale for the dynamic treatment effects and persistence patterns examined in the empirical analysis.

\subsection{Strategic Interaction and Equilibrium Effects}

The analysis thus far treats directors as independent decision-makers who respond exclusively to their own beliefs and incentives. While this assumption is useful for establishing the baseline mechanism, it abstracts from an important feature of organizational behavior. In practice, resignation decisions are often interdependent. The departure of one director may itself constitute information regarding the state of the organization, thereby affecting the incentives of remaining directors.
\\
\\
This possibility is particularly relevant in the context of reputational shocks. Following a major governance scandal or public allegation, stakeholders frequently interpret subsequent resignations as signals regarding latent organizational conditions. A director's departure may therefore increase scrutiny of remaining board members, alter expectations regarding future investigations, or raise concerns regarding information that has not yet become publicly available. In such environments, organizational adjustment may become self-reinforcing.
\\
\\
To capture this mechanism, we introduce strategic complementarities into directors' cost functions. Specifically, suppose that reputational costs depend not only on aggregate reputational pressure and perceived organizational vulnerability, but also on the resignation decisions of peers:

\begin{equation}
C_{ijt}
=
C_{ijt}^{0}
+
\phi_i\rho_t\mathbb{E}_t[\theta_{ijt}]
+
\gamma
\cdot
\frac{1}{J_i}
\sum_{k\neq j}
R_{ikt},
\end{equation}

where $\gamma \geq 0$ measures the strength of strategic interaction.

The final term captures the idea that peer departures generate additional reputational pressure. When $\gamma=0$, directors behave independently and the model collapses to the baseline framework developed above. When $\gamma>0$, the resignation of one director increases the expected cost of remaining for other directors. The parameter therefore measures the extent to which organizational turnover is self-reinforcing.
\\
\\
Several economic interpretations are possible. First, peer departures may reveal information regarding latent organizational conditions. Second, departures may attract additional scrutiny from investors, regulators, employees, or media organizations. Third, resignations may weaken confidence in existing governance structures, increasing uncertainty regarding future organizational stability. Although these mechanisms differ in their institutional details, they share a common implication: departures increase the incentive for additional departures.
\\
\\\\
\\
Substituting equation (13) into the resignation decision yields an equilibrium hazard rate of the form

\begin{equation}
h_{ijt}
=
F
\left(
\phi_i\rho_t\mathbb{E}_t[\theta_{ijt}]
+
\gamma\bar{h}_{it}
-
\kappa_{ijt}
\right),
\end{equation}

where

\[
\bar{h}_{it}
=
\frac{1}{J_i}
\sum_{k\neq j}
h_{ikt}
\]

denotes the average resignation hazard among peers and $\kappa_{ijt}$ collects all remaining determinants of the director's continuation value.
\\
\\
Equation (14) introduces an important feedback mechanism. In the baseline model, aggregate reputational pressure affects resignation behavior directly. In the extended model, reputational pressure affects resignation behavior both directly and indirectly through its effect on the behavior of other directors. Consequently, a common informational shock may generate organizational responses substantially larger than the sum of individual direct effects.
\\
\\
To illustrate this mechanism, consider a symmetric equilibrium in which directors face similar incentives and beliefs. Let

\[
h_{ijt}=h_{it}
\]

for all directors within firm $i$. The equilibrium condition then becomes

\begin{equation}
h_{it}
=
F
\left(
\phi_i\rho_t\mathbb{E}_t[\theta_{it}]
+
\gamma h_{it}
-
\kappa_{it}
\right).
\end{equation}

Differentiating implicitly with respect to aggregate reputational pressure yields

\begin{equation}
\frac{\partial h_{it}}
{\partial \rho_t}
=
\frac{
F'(\cdot)
\phi_i
\mathbb{E}_t[\theta_{it}]
}
{
1-\gamma F'(\cdot)
}.
\end{equation}

Equation (16) highlights the amplification effect generated by strategic interaction. Relative to the baseline model, the response to reputational pressure is multiplied by

\[
\frac{1}
     {1-\gamma F'(\cdot)}.
\]

Whenever

\[
0<\gamma F'(\cdot)<1,
\]

the denominator is less than one and the effect of reputational pressure is amplified. The stronger the strategic complementarity, the larger the organizational response to a given informational shock. The preceding result can be summarized formally.

\begin{proposition}[Organizational Amplification]
Suppose that resignation decisions exhibit strategic complementarities, such that $\gamma>0$. Then the equilibrium response of resignation hazards to an increase in aggregate reputational pressure exceeds the corresponding response in the absence of strategic interaction.
\end{proposition}

\noindent
\textit{Proof.}
Equation (16) implies that

\[
\frac{\partial h_{it}}
{\partial \rho_t}
=
\frac{
F'(\cdot)
\phi_i
\mathbb{E}_t[\theta_{it}]
}
{
1-\gamma F'(\cdot)
}.
\]

Because

\[
0<\gamma F'(\cdot)<1,
\]

the amplification factor

\[
(1-\gamma F'(\cdot))^{-1}
\]

exceeds unity. Therefore, the response to reputational pressure is larger than in the baseline model without strategic interaction. \hfill $\square$
\\
\\
The proposition has an important empirical implication. Organizational turnover should not be interpreted solely as a collection of independent resignation decisions. Instead, departures may trigger additional departures by increasing scrutiny, revealing information, or altering beliefs regarding organizational stability. Consequently, aggregate resignation flows may display persistence, clustering, and nonlinear adjustment following large-scale informational shocks.
\\
\\
This mechanism is particularly relevant in the context of the \#MeToo movement. The initial increase in reputational pressure may induce a subset of directors to resign. Those departures then become additional signals observed by stakeholders and remaining directors, generating further organizational adjustment. The result is a dynamic process in which turnover propagates through the organization over time rather than occurring as a single contemporaneous response.
\\
\\
The strategic interaction mechanism therefore provides a theoretical explanation for two central empirical patterns examined below. First, treatment effects may persist well beyond the initial shock period. Second, resignation responses may exhibit substantial heterogeneity across firms and sectors, reflecting differences not only in exposure but also in the strength of organizational amplification mechanisms.

\subsection{Mapping to the Empirical Specification}

The final step is to connect the theoretical framework to the empirical design. The model implies that firm-level resignation intensity is determined by three fundamental objects: organizational exposure, aggregate reputational pressure, and informational precision. Formally,

\begin{equation}
Y_{it}
=
f(\phi_i,\rho_t,\kappa_t),
\end{equation}

where $\phi_i$ denotes firm-level exposure, $\rho_t$ captures the economy-wide reputational environment, and $\kappa_t$ governs the precision of publicly available information.
\\
\\
The \#MeToo movement generates simultaneous shifts in both $\rho_t$ and $\kappa_t$. As shown above, these changes increase expected reputational costs and alter directors' beliefs regarding latent organizational vulnerability. Because the model is nonlinear, the exact response of resignation intensity depends on the interaction between the shock and firm-level exposure.
\\
\\
To derive an empirically tractable representation, consider a first-order approximation around the pre-shock equilibrium. Let

\[
D_t
=
\mathbf{1}\{t\geq T_0\}
\]

denote an indicator for the post-\#MeToo period. A local approximation implies

\begin{equation}
\Delta Y_{it}
\approx
\beta
\cdot
\phi_i
\cdot
D_t
+
\eta_{it},
\end{equation}

where

\begin{equation}
\beta
=
\mathbb{E}
\left[
\frac{\partial^2 Y_{it}}
{\partial \rho_t \partial \phi_i}
\right]
\end{equation}

summarizes the average exposure-gradient of organizational responses to the shock. The coefficient $\beta$ has a natural economic interpretation. It measures the extent to which the effect of a common increase in reputational pressure depends on organizational exposure. In other words, $\beta$ captures how much more strongly highly exposed firms respond to the same informational shock relative to less exposed firms. The parameter therefore corresponds directly to the central comparative static derived in the previous subsection.
\\
\\
An important implication is that the empirical design differs conceptually from a conventional binary-treatment difference-in-differences framework. The \#MeToo movement affected the entire corporate sector and therefore does not generate a natural untreated control group. Instead, identification arises from differential exposure to a common shock. The relevant comparison is not between treated and untreated firms but between firms characterized by different levels of pre-existing vulnerability to changes in the reputational environment.
\\
\\
The theoretical parameter $\phi_i$ is not directly observable. In the empirical analysis, we therefore replace $\phi_i$ with an observable proxy, denoted $E_i$, constructed using Item 5.02 SEC Form 8-K disclosure intensity measured prior to October 2017. Because the exposure measure is calculated entirely from pre-shock disclosures, it captures predetermined organizational characteristics rather than post-shock responses. Substituting the empirical proxy for the latent exposure parameter yields

\begin{equation}
\Delta Y_{it}
\approx
\beta
\left(
E_i
\cdot
D_t
\right)
+
\eta_{it},
\end{equation}

which corresponds directly to the baseline empirical specification estimated below. This mapping provides a structural interpretation of the empirical coefficient. The interaction term does not merely capture differential turnover across firms. Rather, it estimates the extent to which pre-existing organizational exposure amplifies the effect of a common reputational-information shock on resignation behavior. Consequently, the empirical specification can be interpreted as a reduced-form approximation to the comparative statics generated by the theoretical model.
\\
\\
The mapping also clarifies why the empirical analysis emphasizes dynamic treatment effects. Because belief updating, organizational learning, internal investigations, and strategic complementarities unfold over time, the model does not predict that all adjustment should occur immediately following October 2017. Instead, exposure-sensitive responses may emerge gradually as directors revise beliefs and reassess the expected costs of continued affiliation. This prediction motivates the dynamic event-study and matrix-completion analyses presented below.

\subsection{Interpretation and Empirical Implications}

The theoretical framework developed above provides a structural interpretation of the empirical specification and clarifies the economic mechanism underlying the estimated effects. In particular, the model implies that organizational turnover is determined by the interaction between three fundamental forces: firm-level exposure to governance-sensitive reputational risk, aggregate reputational pressure, and directors' beliefs regarding latent organizational vulnerability. The \#MeToo movement affects behavior by simultaneously increasing the reputational price attached to adverse information and increasing the precision with which such information is interpreted.
\\
\\
A key implication is that the empirical coefficient of interest has a natural structural interpretation. The interaction between exposure and the post-\#MeToo period captures the cross-partial derivative

\[
\frac{\partial^2 Y_{it}}
{\partial \rho_t \partial \phi_i},
\]

which measures how strongly organizational turnover responds to a common increase in reputational pressure as firm-level exposure rises. The coefficient therefore does not simply reflect higher turnover among exposed firms. Rather, it measures the extent to which exposure amplifies the organizational response to an economy-wide reputational-information shock. In this sense, the empirical specification can be interpreted as a reduced-form approximation to the central comparative static generated by the model.
\\
\\
The framework also clarifies why a continuous-treatment design is the appropriate empirical strategy. Unlike conventional policy interventions, the \#MeToo movement did not affect a subset of firms while leaving others untreated. The shock altered the informational environment facing the entire corporate sector. Consequently, identification cannot rely on a simple treated-versus-control comparison. Instead, the relevant source of variation arises from differential exposure to a common shock. Firms characterized by greater ex ante vulnerability should experience larger increases in resignation intensity than firms characterized by lower exposure. The empirical analysis therefore estimates an exposure gradient rather than an average treatment effect in the conventional binary-treatment sense.
\\
\\
The model further explains why first differences isolate the relevant behavioral response. Organizational turnover reflects both persistent firm-specific characteristics and time-varying responses to changes in the informational environment. Differencing removes time-invariant organizational features and focuses attention on changes in resignation behavior associated with the shift in reputational pressure induced by the \#MeToo movement. The resulting estimates therefore capture adjustment dynamics rather than permanent differences in governance structures across firms.
\\
\\
Perhaps most importantly, the framework implies that observed empirical effects need not correspond solely to direct responses to the initial shock. The combination of Bayesian belief updating and strategic complementarities generates endogenous amplification. Directors revise beliefs as information becomes more salient and more informative, while peer departures create additional incentives for organizational exit. Consequently, observed turnover reflects both the direct effect of increased reputational pressure and the indirect effect of organizational feedback mechanisms. This distinction is important because it implies that large and persistent resignation flows need not require large initial shocks. Even modest changes in beliefs may generate substantial organizational adjustment when amplified through endogenous responses within the firm's governance structure.
\\
\\
The theory also provides a natural explanation for heterogeneous treatment effects across firms and sectors. Exposure varies across organizations, informational environments differ across industries, and strategic complementarities may be stronger in some governance structures than in others. As a consequence, identical changes in aggregate reputational pressure need not generate identical organizational responses. The model therefore predicts substantial variation in resignation dynamics across firms, industries, and organizational environments, a prediction examined directly in the empirical analysis.
\\
\\
Taken together, the framework generates four testable predictions. First, firms characterized by greater pre-shock exposure should experience larger increases in resignation intensity following the \#MeToo shock. Second, treatment effects should persist over time because organizational learning, belief updating, and internal investigations unfold gradually. Third, effects should appear across multiple organizational margins, including voluntary exits, forced departures, board turnover, executive turnover, and allegation-related disclosures. Fourth, responses should exhibit substantial heterogeneity across sectors and organizational environments, reflecting differences in exposure, reputational sensitivity, and the strength of strategic interaction.
\\
\\
These predictions provide the foundation for the empirical analysis that follows. The results section therefore does not simply test whether the \#MeToo movement affected organizational turnover. Rather, it evaluates whether the patterns of organizational adjustment observed in the data correspond to the behavioral mechanisms implied by the theoretical framework. In this sense, the empirical estimates can be interpreted as evidence regarding the underlying process through which large-scale reputational-information shocks propagate through corporate governance structures and influence executive and board turnover.

\section{Empirical Strategy and Identification}

The theoretical framework developed in the previous section implies that the \#MeToo movement operated as a common reputational-information shock whose effects varied systematically with firms' pre-existing exposure to governance-sensitive reputational risk. This prediction motivates the empirical strategy adopted in the paper. Because the shock affected the entire corporate sector simultaneously, identification cannot rely on a conventional treated-versus-control comparison. Instead, the analysis exploits heterogeneity in pre-shock exposure and examines whether firms characterized by greater vulnerability experienced larger increases in resignation intensity following October 2017.
\\
\\
To evaluate this prediction, we employ a sequence of complementary estimators. We begin with a continuous-treatment difference-in-differences design that directly corresponds to the comparative statics derived in Section 3. We then examine dynamic adjustment using event-study methods, reconstruct counterfactual turnover trajectories through matrix completion, and assess statistical significance using both clustered and wild-bootstrap inference procedures. Taken together, these approaches provide a comprehensive evaluation of the organizational consequences of the \#MeToo movement under increasingly flexible assumptions regarding the evolution of untreated outcomes.

\subsection{Identification under a Common Reputational Shock}

The central empirical challenge arises from the nature of the treatment itself. Unlike conventional policy interventions that affect a subset of units while leaving others untreated, the \#MeToo movement constituted a common informational and reputational shock that affected the entire corporate sector simultaneously. Beginning with the October 2017 Weinstein revelations, the informational environment surrounding corporate governance changed abruptly. Allegations of misconduct became more salient, stakeholder scrutiny intensified, and the expected reputational cost of governance failures increased. Because this shock was economy-wide, there exists no natural untreated group against which exposed firms can be compared.
\\
\\
The theoretical framework developed in the previous section provides a solution to this identification problem. The model predicts that firms differ systematically in their sensitivity to common reputational shocks through the exposure parameter $\phi_i$. Although all firms experienced the same increase in reputational pressure, highly exposed firms should exhibit larger increases in resignation intensity than firms characterized by lower exposure. Identification therefore derives from heterogeneous exposure to a common shock rather than from differential treatment assignment.
\\
\\
This distinction is important. The empirical objective is not to estimate whether the \#MeToo movement affected corporate governance in general. Given the economy-wide nature of the shock, such a comparison is impossible. Instead, the objective is to estimate whether firms characterized by greater ex ante exposure experienced systematically larger changes in resignation behavior following the shock. The resulting estimand is therefore an exposure gradient rather than a conventional average treatment effect.
\\
\\
Let $Y_{it}$ denote firm-level resignation intensity. Consistent with the theoretical framework, the outcome of interest is the change in resignation behavior rather than the level of resignations. Firms differ substantially in baseline turnover rates due to persistent organizational characteristics unrelated to the \#MeToo movement. Focusing on changes in resignation intensity isolates organizational adjustment associated with the shift in the informational environment while reducing the influence of persistent cross-sectional heterogeneity.
\\
\\
Formally, define

\begin{equation}
\Delta Y_{it}=Y_{it}-Y_{i,t-1}.
\end{equation}

The exposure measure is constructed entirely from information observed before October 2017. Specifically, let $8K502_{it}$ denote governance-related Item 5.02 Form 8-K disclosures and define

\begin{equation}
E_i=
\frac{1}{T_0-1}
\sum_{t<T_0}
8K502_{it}.
\end{equation}

By construction, $E_i$ is predetermined with respect to the \#MeToo shock and therefore cannot be affected by post-treatment organizational responses.
\\
\\
The theoretical model implies that the causal effect of the shock should be increasing in exposure:

\begin{equation}
\frac{\partial}
{\partial E_i}
\mathbb{E}
\left[
\Delta Y_{it}(1)
-
\Delta Y_{it}(0)
\right]
>0.
\end{equation}

The empirical analysis below evaluates this prediction using a sequence of estimators that progressively relax identifying assumptions.

\subsection{Baseline Continuous-Treatment Difference-in-Differences}

The theoretical framework predicts that firms characterized by greater exposure should experience larger increases in resignation intensity following the \#MeToo shock. A natural starting point is therefore a continuous-treatment difference-in-differences framework.
\\
\\
Let $Y_{it}$ denote resignation intensity measured as the logarithm of one plus the total number of board and executive departures in firm $i$ during month $t$. The logarithmic transformation is motivated by both theoretical and empirical considerations. From a theoretical perspective, the model predicts changes in the flow of resignations rather than a binary departure event. From an empirical perspective, resignation counts are highly skewed and characterized by occasional episodes of clustered turnover.
\\
\\
The baseline model specification is

\begin{equation}
Y_{it}
=
\alpha_i
+
\lambda_t
+
\beta
(E_i \times D_t)
+
\mathbf{X}_{it}'\Gamma
+
\varepsilon_{it},
\end{equation}

where $\alpha_i$ denotes firm fixed effects, $\lambda_t$ denotes month fixed effects, $E_i$ represents pre-shock exposure, and $D_t$ is an indicator equal to one after October 2017.
\\
\\
The coefficient of interest is $\beta$, which captures the differential response of resignation intensity to the \#MeToo shock as exposure increases. The specification follows directly from the comparative statics developed in Section 3:

\begin{equation}
\frac{\partial^2Y_{it}}
{\partial\rho_t\partial\phi_i}
>0.
\end{equation}

Firm fixed effects absorb time-invariant organizational characteristics, while month fixed effects absorb aggregate shocks affecting all firms simultaneously. Identification therefore relies on differential changes in resignation behavior across firms with different levels of predetermined exposure.
\\
\\
To account for persistence in organizational turnover, the preferred specification augments the baseline model with lagged resignation intensity:

\begin{equation}
Y_{it}
=
\alpha_i
+
\lambda_t
+
\beta(E_i \times D_t)
+
\sum_{k=1}^{4}
\rho_kY_{i,t-k}
+
\varepsilon_{it}.
\end{equation}

The validity of this design relies on a generalized parallel-trends assumption under which firms characterized by different levels of exposure would have followed similar turnover trajectories in the absence of the shock.

\subsection{Dynamic Event-Study Analysis}

While the baseline specification provides an estimate of the average organizational response, it cannot reveal how that response evolves over time. The theoretical framework predicts dynamic adjustment arising from Bayesian learning, internal investigations, governance adaptation, and strategic complementarities.
\\
\\
To examine these dynamics, we estimate

\begin{equation}
Y_{it}
=
\alpha_i
+
\lambda_t
+
\sum_{k\neq -1}
\beta_k
\left(
E_i
\times
\mathbf{1}\{t-T_0=k\}
\right)
+
\sum_{p=1}^{4}
\rho_pY_{i,t-p}
+
\varepsilon_{it}.
\end{equation}

The coefficients $\beta_k$ trace the evolution of resignation intensity before and after the shock as a function of exposure. The event-study framework serves two purposes. First, it provides a diagnostic assessment of the identifying assumptions underlying the continuous-treatment design. Specifically, the generalized parallel-trends assumption implies

\begin{equation}
\beta_k=0
\qquad
\forall k<0.
\end{equation}

Second, the event-study provides a direct test of the dynamic predictions generated by the theoretical framework. Bayesian updating and organizational amplification imply that treatment effects should persist beyond the initial shock period rather than dissipating immediately.

\subsection{Dynamic Matrix Completion}

Although the difference-in-differences and event-study specifications provide useful evidence, both ultimately rely on assumptions regarding the evolution of untreated outcomes. To relax these assumptions, we complement the baseline analysis with dynamic matrix completion methods.
\\
\\
Following \citet{athey2021matrix}, outcomes are assumed to satisfy

\begin{equation}
Y_{it}
=
L_{it}
+
\tau_{it}W_{it}
+
u_{it},
\end{equation}

where $L_{it}$ denotes the latent untreated outcome process and $\tau_{it}$ denotes the treatment effect. The key identifying assumption is that untreated outcomes are generated by a low-rank latent structure:

\begin{equation}
\text{rank}(L)
\ll
\min(N,T).
\end{equation}

Matrix completion estimates the latent outcome matrix by solving

\begin{equation}
\widehat{L}
=
\arg\min_{L}
\left\{
\sum_{(i,t)\in\Omega}
(Y_{it}-L_{it})^2
+
\lambda\|L\|_*
\right\},
\end{equation}

where $\|L\|_*$ denotes the nuclear norm. The estimator is particularly attractive in the present setting because it allows untreated outcomes to evolve flexibly while accommodating latent time-varying confounders. This feature is especially important when analyzing organizational outcomes such as executive turnover, which may be influenced by unobserved governance processes and managerial dynamics.
\\
\\
Dynamic treatment effects are computed as

\begin{equation}
\widehat{\tau}_k
=
\frac{1}{N_T}
\sum_{i\in T}
\left[
Y_{i,T_0+k}
-
\widehat{Y}_{i,T_0+k}(0)
\right].
\end{equation}

These estimates provide a direct assessment of whether the persistent adjustment patterns predicted by the theoretical framework are supported by the data.

\subsection{Statistical Inference}

A central concern in panel-data settings involving organizational outcomes is the reliability of statistical inference. Executive and board turnover decisions frequently exhibit persistence, clustering, and serial dependence, potentially causing conventional standard errors to understate uncertainty. To address these concerns, all baseline specifications employ firm-clustered standard errors. Clustering at the firm level permits arbitrary forms of heteroskedasticity and serial correlation within firms while maintaining independence across firms.

The cluster-robust variance estimator is

\begin{equation}
\widehat{\mathrm{Var}}(\widehat{\beta})
=
(X'X)^{-1}
\left(
\sum_{g=1}^{G}
X_g'
\widehat{u}_g
\widehat{u}_g'
X_g
\right)
(X'X)^{-1}.
\end{equation}

Because recent research has demonstrated that conventional clustered inference may perform poorly in finite samples, we additionally employ wild-cluster bootstrap procedures. These methods generate empirical sampling distributions through repeated cluster-level resampling and provide more reliable inference under residual dependence.
\\
\\
Bootstrap inference is implemented using multiple replication schemes, including 100, 200, 500, and 1,000 replications. Examining a range of bootstrap procedures allows assessment of the stability of estimated sampling distributions and ensures that statistical significance is not driven by a particular implementation.
\\
\\
Throughout the paper, statistical significance is therefore evaluated using a hierarchy of increasingly conservative procedures. The consistency of conclusions across clustered and bootstrap-based inference provides additional confidence that the estimated effects reflect genuine organizational responses rather than finite-sample artifacts or misspecified variance structures.

\section{Data and Samples}

Our analysis based on an original, hand-collected dataset of leadership-change disclosures filed by S\&P 500 firms between 2016 and 2023. The choice of source follows from institutional design. Under SEC rules, departures of directors and senior officers are typically reported through Item 5.02 of Form 8-K, which provides a standardized disclosure channel for leadership turnover across public companies. At the same time, Item 5.02 disclosures are substantively sparse: firms usually disclose the fact and date of a departure, but frequently omit the underlying reason or describe it in vague terms such as ``personal reasons.'' The combination of formal standardization and substantive opacity makes Item 5.02 filings particularly well-suited for studying changes in resignation behavior even when the precise reason for a departure remains hidden.
\\
\\
The raw data consist of all relevant Item 5.02 filings submitted by S\&P 500 firms during the sample period. In the broader departure corpus assembled from the filings, there are 3,841 observations involving director and senior officer departures. The relevant unit of observation is the firm-month bracket, and the outcome of interest is not the level of director and officer turnover itself, but the change in departure intensity. This follows directly from the identification strategy that we describe in detail below. Since the Weinstein saga in October 2017 constituted a common, economy-wide reputational shock rather than a treatment affecting only a subset of firms, identification cannot come from discrete comparison between treatment and control, but instead from whether firms with greater pre-shock exposure indeed experienced larger post-shock changes in departure flows than firms with lower exposure.  
\\
\\
Our key explanatory variable is firm-level exposure to governance-sensitive disclosure environments. Following the identification strategy, we measure exposure using pre-shock frequency of Item 5.02 filings. For each firm, the exposure variable is defined as the average monthly frequency of Item 5.02 filings before October 2017. Again, firms with higher pre-shock filing intensity are interpreted as firms whose governance structure were more exposed to disclosure-related reputational pressure, and hence as firms more likely to respond strongly when the informational environment shifted exogenously after October 2017. The empirical design compared firms with different values of exposure before and after the shock. This is implemented through the interaction of a post-\#MeToo indicator with the continuous exposure measure.

\subsection{Measuring Organizational Adjustment: Cumulative Resignations and Resignation Intensity}

A central empirical issue concerns the appropriate measurement of organizational turnover following the \#MeToo shock. Executive departures are relatively infrequent events and exhibit substantial heterogeneity across firms.\footnote{It should be noted that the analysis applies to any official, as well as to the firm's decision to terminate, rather than the director's decision to resign. Consequently, we use ``director'' and ``resignation'' terms for the sake of simplicity.} While many firms experience no observable turnover during a given period, others undergo repeated episodes of managerial restructuring involving multiple departures across different levels of the organizational hierarchy. As shown in the descriptive statistics, the distribution of resignation activity is highly skewed and characterized by substantial dispersion across firms and time.
\\
\\
To capture this variation, the analysis employs cumulative resignation activity transformed into logarithmic form as the primary outcome variable. Specifically, the dependent variable is defined as:
\[
Y_{it}=\log\left(1+\text{CumResignations}_{it}\right),
\]
where \(\text{CumResignations}_{it}\) denotes cumulative resignation activity for firm \(i\) in month \(t\), and one is added to preserve observations with zero resignation events.
\\
\\
Several considerations motivate this specification. First, cumulative resignation activity captures the broader process of organizational adjustment rather than isolated executive departures. Organizational restructuring often unfolds sequentially. The departure of one executive may trigger subsequent exits, board changes, internal investigations, or broader revisions in governance structures. Focusing exclusively on individual resignation events would therefore risk understating the extent of organizational responses following a major informational shock.
\\
\\
Second, the logarithmic transformation provides a natural measure of resignation intensity rather than absolute turnover counts. Raw resignation frequencies are highly right-skewed, with a relatively small number of firms experiencing unusually large numbers of departures. Estimating effects in levels would therefore allow a small subset of firms to exert disproportionate influence on parameter estimates. The logarithmic transformation reduces the influence of extreme observations while preserving meaningful variation across firms.
\\
\\
Third, the use of log resignation intensity permits an economically intuitive interpretation of estimated coefficients. Under this specification, estimated treatment effects approximately correspond to percentage changes in organizational turnover intensity rather than absolute changes in resignation counts. Such interpretations are particularly useful when comparing organizational responses across firms with different baseline turnover profiles and governance structures.
\\
\\
Finally, the dynamic framework employed throughout the paper incorporates lagged resignation intensity to account for persistence in organizational adjustment processes. Consequently, estimated treatment effects should not be interpreted as shifts in steady-state turnover levels. Instead, the coefficients capture differential changes in the intensity of organizational restructuring among firms with greater pre-treatment exposure following the \#MeToo shock. Taken together, this approach allows the empirical analysis to focus on the central object of interest, the extent to which the \#MeToo movement altered the pace and magnitude of organizational adjustment rather than simply increasing the probability of isolated executive departures.

\subsection{Item 5.02 Filings and Textual Analysis}

A central contribution of this study lies in the construction of a novel firm-level exposure index based on large-scale computational textual analysis of mandatory corporate disclosure filings. The empirical framework builds on a rapidly expanding literature combining advances in computational social science, machine learning, and natural language processing with causal inference in economics and finance. Recent work has increasingly emphasized the value of high-dimensional textual data for measuring latent institutional characteristics, organizational expectations, reputational risk, and governance quality that are difficult to observe directly through conventional accounting variables or survey-based indicators. Important contributions in this literature include \citep{gentzkow2019text} on text as data methods in economics, \citep{hoberg2016text} on textual product-market similarity, \citep{hassan2019political} on geopolitical risk extraction from corporate disclosures, \citep{baker2016uncertainty} on textual uncertainty measurement, and \cite{tetlock2007sentiment, tetlock2008fundamentals} on financial text and market behavior. The present study extends this broader research agenda by constructing a dynamic firm-level measure of organizational exposure to reputational and governance vulnerability using pre-treatment SEC disclosure text.
\\
\\
Unlike existing studies that rely on ex post classifications of scandal involvement, media coverage, or realized allegations, the empirical design developed here measures firms' ex ante exposure to reputational and governance vulnerability prior to the emergence of the \#MeToo shock in October 2017. The resulting measure captures latent organizational susceptibility to reputational disruption embedded in the informational structure of corporate disclosures before the Weinstein revelations transformed the equilibrium salience of workplace misconduct, governance accountability, and executive oversight. Unlike existing studies that rely on ex post classifications of scandal involvement, media coverage, or realized allegations, the empirical design developed here measures firms' ex ante exposure to reputational and governance vulnerability prior to the emergence of the \#MeToo shock in October 2017. The resulting measure captures latent organizational susceptibility to reputational disruption embedded in the informational structure of corporate disclosures before the Weinstein revelations transformed the equilibrium salience of workplace misconduct, governance accountability, and executive oversight.
\\
\\
The underlying dataset was assembled through a large-scale computer-assisted collection and coding exercise spanning all S\&P 500 firms between January 2016 and July 2024. The resulting data architecture combines methods from computational linguistics, natural language processing, information retrieval, and panel econometrics in order to transform unstructured disclosure text into structured measures of organizational exposure and governance vulnerability. In contrast to manually coded event-study datasets commonly used in corporate governance research, the present framework relies on automated parsing and classification procedures capable of processing large volumes of heterogeneous disclosure text at high temporal frequency. The raw data collection process required the systematic retrieval, parsing, classification, and organization of thousands of mandatory SEC Form 8-K disclosures filed by publicly listed corporations. Particular emphasis was placed on Item 5.02 filings, which contain mandatory disclosures regarding departures, appointments, resignations, removals, or transitions involving directors and senior executive officers. Item 5.02 disclosures occupy a particularly important institutional position within the architecture of U.S. securities regulation because they require firms to disclose material changes in executive leadership and board composition within a legally mandated reporting framework overseen by the Securities and Exchange Commission. The reporting requirement therefore creates a relatively standardized and high-frequency disclosure environment through which organizational restructuring, executive departures, governance disputes, monitoring failures, and leadership transitions become observable in near real time.
\\
\\
From an institutional perspective, Item 5.02 filings are especially valuable for empirical analysis because they sit at the intersection of corporate governance, securities regulation, and organizational accountability. Unlike voluntary press releases or media interviews, these disclosures are filed under formal SEC reporting obligations and therefore impose substantially stronger constraints on strategic omission or selective disclosure. The filings frequently contain detailed narrative explanations regarding the circumstances surrounding executive transitions, including references to disagreements with management, compliance issues, governance concerns, internal investigations, organizational restructuring, and board decisions. As a result, the textual environment embedded within Item 5.02 disclosures provides unusually rich information regarding the internal governance structure and organizational dynamics of publicly traded corporations.
\\
\\
The institutional role of Item 5.02 filings became particularly important following the emergence of the \#MeToo movement in October 2017. As governance scrutiny intensified and reputational accountability mechanisms strengthened, these filings increasingly served as formal disclosure channels through which firms communicated executive departures, leadership restructuring, and organizational responses to heightened public scrutiny. Consequently, the Item 5.02 disclosure environment provides a uniquely appropriate empirical setting for studying how reputational-information shocks propagate through corporate governance systems and translate into observable organizational adjustment. These filings are especially valuable because they represent legally mandated disclosures rather than voluntary public-relations communications, thereby substantially reducing concerns regarding selective disclosure or endogenous media amplification.
\\
\\
The data collection architecture combined automated textual extraction procedures with extensive human-assisted verification and classification. At the computational level, the pipeline involved document retrieval, textual normalization, tokenization, phrase extraction, contextual classification, semantic filtering, and metadata harmonization across filings and firms. The resulting corpus enabled the identification of textual patterns associated with executive turnover, internal organizational instability, governance restructuring, misconduct-related disclosures, reputational concerns, compliance language, legal-risk signaling, and managerial transitions. Particular attention was devoted to minimizing false-positive classifications generated by routine executive succession events unrelated to organizational disruption.
\\
\\
The exposure index was constructed using a high-dimensional textual environment derived from pre-treatment disclosure language. The classification framework was designed to capture not merely explicit references to misconduct, but broader latent dimensions of governance fragility and reputational sensitivity embedded within disclosure narratives. In this sense, the exposure index functions analogously to latent-feature extraction procedures increasingly employed in computational finance and machine-learning applications, where textual signals are used to recover otherwise unobservable institutional characteristics. Rather than treating disclosure text as unstructured narrative noise, the framework interprets the disclosure environment itself as an informational representation of organizational vulnerability. The initial stage involved automated retrieval and parsing of firm-level 8-K filings from SEC disclosure repositories. Each filing was subsequently processed using computer-coded textual analysis procedures designed to identify references to executive turnover, governance restructuring, allegations, reputational risk, organizational misconduct, compliance concerns, and related dimensions of corporate vulnerability. The coding framework was designed to distinguish ordinary executive transitions from organizational events plausibly associated with governance disruption, reputational pressure, monitoring failures, or internal institutional instability.
\\
\\
A particularly important feature of the dataset is that the exposure index was constructed exclusively from filings submitted before October 2017. This timing restriction is critical for identification. Because the textual exposure measure is entirely predetermined relative to the Weinstein revelations and the subsequent emergence of the \#MeToo movement, the index cannot mechanically reflect post-treatment organizational responses, resignations, allegations, or disclosure adjustments induced by the shock itself. Instead, the exposure index captures pre-existing latent susceptibility to reputational and governance disruption embedded in firms' disclosure histories prior to the shock.
\\
\\
Formally, let \begin{math} Exposure_i\end{math} denote the firm-level exposure index constructed from the corpus of pre-treatment 8-K disclosures. The index aggregates information extracted from the textual environment of firm disclosures submitted before October 2017 and therefore reflects ex ante governance and reputational vulnerability rather than realized post-shock organizational outcomes. The empirical strategy subsequently exploits cross-sectional variation in this predetermined exposure measure by interacting it with the post-Weinstein period indicator:

\begin{equation}
    Treatment_{it}=Exposure_i\times Post_t
\end{equation}

where \begin{math} Post_t\end{math} equals one for months following October 2017 and zero otherwise. Under this framework, identification derives not from differential treatment timing but from heterogeneous exposure to a common reputational-information shock. The resulting design differs fundamentally from conventional event-study approaches examining media scandals or corporate controversies. Methodologically, the framework is conceptually closer to recent work in computational economics that combines high-dimensional text representations with causal identification strategies in panel settings. The central empirical objective is not merely to identify whether resignations increased following October 2017, but rather to estimate how a large-scale reputational-information shock propagated through heterogeneous organizational structures characterized by different latent exposure profiles prior to the shock. This distinction is crucial because it shifts the analysis away from contemporaneous scandal realization and toward ex ante informational susceptibility.
\\
\\
From a computational perspective, the framework effectively converts large-scale unstructured disclosure text into a firm-level treatment-intensity measure that can subsequently be embedded within dynamic panel estimators. This architecture substantially expands the dimensionality of information available for empirical corporate-governance analysis relative to traditional approaches relying exclusively on accounting ratios, governance indices, or manually coded scandal indicators. In this sense, the paper contributes not only to the economics of organizational adjustment and reputational shocks, but also to the broader integration of computational methods into empirical political economy and corporate finance. Rather than comparing directly implicated firms with untreated firms after allegations become public, the present framework identifies whether firms characterized by greater ex ante textual exposure experienced systematically different post-shock organizational adjustment dynamics relative to firms with lower exposure. In this sense, the design is conceptually closer to models of heterogeneous susceptibility to common informational shocks, reputational contagion, or institutional salience shocks than to standard binary-treatment policy evaluations.
\\
\\
The scale and granularity of the underlying data collection effort substantially strengthen the credibility of the empirical analysis. The database contains detailed information on executive and board resignations across firms, sectors, states, and executive positions, including board members, vice presidents, chief financial officers, chairpersons, and female executives. Moreover, the coding structure permits the distinction between voluntary resignations, forced departures, governance restructuring, allegation-related disclosures, and broader organizational turnover dynamics. This degree of institutional detail is rarely available in existing studies of corporate governance responses to reputational shocks.
\\
\\
The resulting dataset provides several important advantages. First, the use of mandatory SEC disclosures substantially mitigates concerns regarding media-selection bias. Second, the reliance on pre-treatment textual information reduces concerns regarding reverse causality and post-treatment contamination. Third, the panel structure permits dynamic estimation of organizational adjustment processes over time, allowing the analysis to distinguish temporary disruption from persistent governance restructuring. Fourth, the combination of textual exposure measurement and high-frequency resignation outcomes allows the empirical framework to identify multiple margins of organizational response, including disciplinary governance, anticipatory reputational exit, monitoring restructuring, compliance adaptation, and informational amplification.
\\
\\
The richness of the disclosure corpus also enables a more nuanced interpretation of the \#MeToo shock itself. The empirical evidence suggests that the movement did not affect all firms symmetrically. Instead, the shock propagated through pre-existing informational and organizational structures embedded within firms' disclosure environments. Firms exhibiting greater ex ante textual exposure experienced systematically larger increases in executive and board turnover following October 2017, consistent with the interpretation that \#MeToo functioned as a large-scale reputational-information shock interacting with latent organizational vulnerability.
\\
\\
Importantly, the exposure index should not be interpreted as a direct measure of realized misconduct. Rather, it captures broader dimensions of governance fragility, reputational sensitivity, disclosure intensity, organizational instability, and institutional susceptibility embedded in the informational architecture of firm disclosures prior to the shock. This distinction is central because it implies that the estimated treatment effects reflect heterogeneous organizational exposure to reputational pressure rather than merely the consequences of publicly revealed allegations.
\\
\\
Taken together, the large-scale textual data collection effort, the construction of a predetermined exposure index, and the dynamic firm-level panel structure provide a uniquely rich empirical environment for studying how reputational-information shocks propagate through corporate governance systems. The resulting framework allows the analysis to move beyond descriptive correlations and instead identify the mechanisms through which exogenous shifts in social accountability norms reshape organizational behavior, executive turnover, and internal governance structures. Figure 1 presents state-level distribution of the 5.08 filings before the New York Post revelation of the Weinstein case in October 2017 whereas sector-specific distribution of the average case filings id presented in Figure 2.

\begin{figure}
    \centering
    \includegraphics[width=1.1\linewidth]{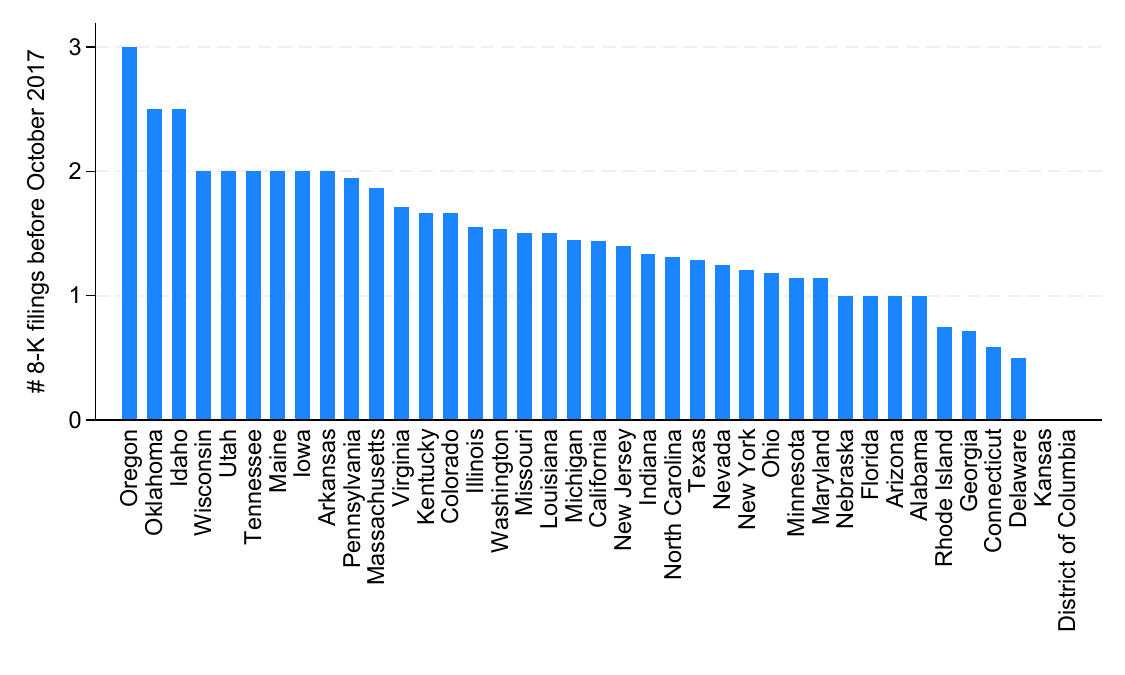}
    \caption{Item 5.08 Filings Before Weinstein Case Revelation Across States}
    \label{fig:placeholder}
\end{figure}

\begin{figure}
    \centering
    \includegraphics[width=1.1\linewidth]{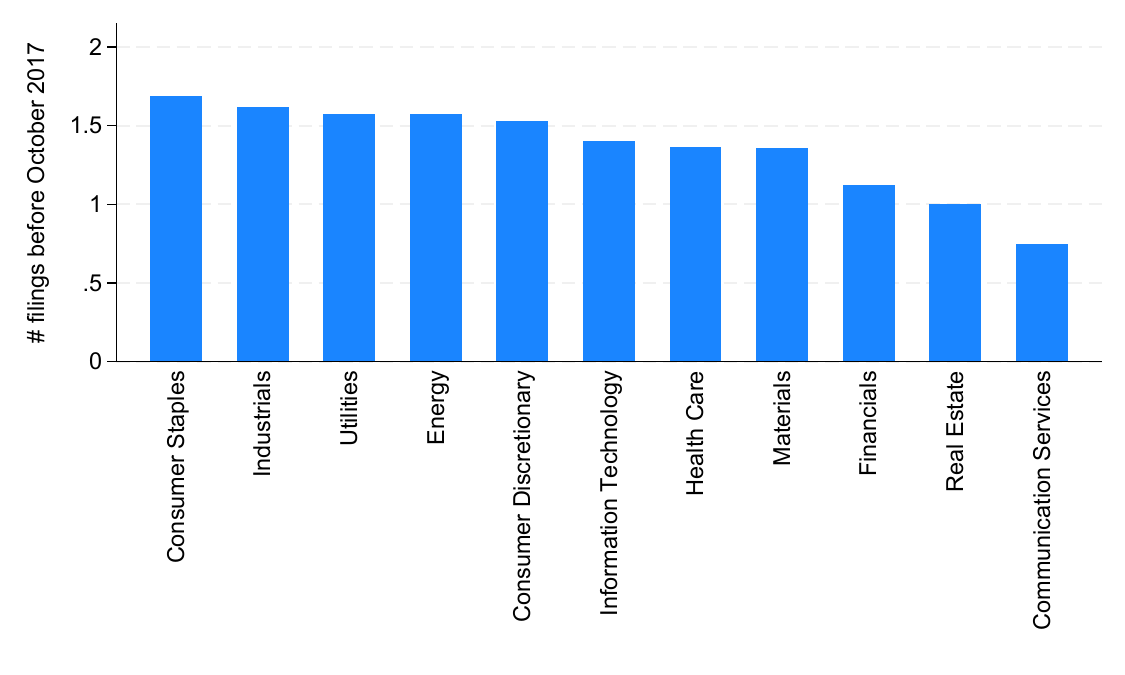}
    \caption{Item 5.08 Filings Before Weinstein Case Across Sectors}
    \label{fig:placeholder}
\end{figure}

\subsection{Descriptive Statistics}

Table 1 reports descriptive statistics for the firm-month panel spanning January 2016 through July 2024. The final sample contains 51,791 firm-month observations across S\&P 500 corporations and captures detailed organizational changes in executive and board composition derived from mandatory SEC disclosure filings. The panel structure combines high-frequency temporal variation with rich institutional detail regarding the type and nature of executive turnover events.
\\
\\
Several features of the data deserve attention. First, cumulative resignation activity exhibits substantial variation across firms and over time. The average firm records approximately 3.92 cumulative resignation events over the sample period, with a standard deviation of 4.31 and a range extending from zero to forty-six events. The large dispersion relative to the mean indicates considerable heterogeneity in organizational turnover across firms. While some firms experience almost no observable executive transitions during the sample period, others undergo repeated episodes of leadership restructuring. Such variation is economically important because it suggests that executive turnover is not uniformly distributed across firms but instead may reflect underlying differences in governance structures, organizational stability, and exposure to reputational or managerial risk.

\begin{table}[!htbp]
\centering
\caption{Summary Statistics of Executive Turnover and Organizational Characteristics}
\label{tab:summary_stats}

\scriptsize
\setlength{\tabcolsep}{4pt}
\renewcommand{\arraystretch}{1.05}

\begin{threeparttable}

\resizebox{0.98\textwidth}{!}{
\begin{tabular}{lccccccc}
\toprule
& \# Obs. & \multicolumn{3}{c}{Mean} & Std. Dev. & Min & Max \\
\cmidrule(lr){3-5}
Variable &  & Overall & Before & After &  &  &  \\
\midrule
Resignation Intensity   & 51,791 & 3.918 & 0.755 & 4.791 & 4.309 & 0 & 46 \\
Voluntary               & 51,791 & 0.053 & 0.047 & 0.054 & 0.230 & 0 & 3 \\
Forced                  & 51,791 & 0.002 & 0.001 & 0.002 & 0.041 & 0 & 2 \\
Female Executive        & 51,791 & 0.005 & 0.004 & 0.005 & 0.074 & 0 & 1 \\
Board Member            & 51,791 & 0.021 & 0.017 & 0.022 & 0.146 & 0 & 2 \\
Vice President          & 51,791 & 0.016 & 0.014 & 0.016 & 0.127 & 0 & 2 \\
Chief Financial Officer & 51,791 & 0.006 & 0.006 & 0.006 & 0.081 & 0 & 2 \\
Flagged Allegation      & 51,791 & 0.002 & 0.0009 & 0.002 & 0.043 & 0 & 1 \\
\bottomrule
\end{tabular}
}

\vspace{0.15cm}

\begin{minipage}{0.98\textwidth}
\scriptsize
\raggedright
\textit{Notes:} The table reports descriptive statistics for the main variables used in the analysis. ``Before'' refers to the pre-\#MeToo period prior to October 2017, while ``After'' corresponds to the post-Weinstein period. Resignation intensity measures cumulative executive and board turnover events at the firm-month level. Voluntary and forced departures distinguish between self-initiated and board-enforced executive exits. Additional variables capture turnover among female executives, board members, vice presidents, and chief financial officers, as well as disclosure activity related to flagged allegations. The sample consists of monthly observations for publicly listed U.S. firms from January 2016 through July 2024.
\end{minipage}

\vspace{0.15cm}
\hrule

\end{threeparttable}
\end{table}

More importantly, average cumulative resignation activity increases markedly after October 2017. Mean resignation intensity rises from 0.755 in the pre-Weinstein period to 4.791 thereafter, suggesting a substantial increase in executive turnover following the emergence of the \#MeToo movement. While these raw differences are not themselves causal estimates and may reflect broader secular trends, they provide preliminary evidence consistent with a structural shift in organizational dynamics during the post-shock period.
\\
\\
Second, the decomposition of turnover events reveals considerable heterogeneity across resignation types. Voluntary resignations represent the most common executive transition outcome, with a mean monthly frequency of 0.053 events per firm. In contrast, forced departures are relatively rare, averaging only 0.002 events per firm-month. This asymmetry is consistent with prior corporate-governance evidence suggesting that involuntary executive removal constitutes an exceptional rather than routine organizational event. Forced turnover often reflects more severe episodes of governance intervention, board discipline, or organizational distress.
\\
\\
Third, turnover among specific executive positions also displays meaningful variation. Board-member departures occur at an average rate of 0.021 per firm-month, exceeding rates observed for vice presidents (0.016) and chief financial officers (0.006). The relatively greater frequency of board turnover suggests that organizational adjustment may operate not only through managerial exits but also through restructuring of monitoring institutions themselves. Since boards occupy a central role in oversight and accountability, changes in board composition potentially reflect responses to shifts in governance expectations and reputational pressures.
\\
\\
Fourth, female executive departures are comparatively infrequent, averaging 0.005 events per firm-month. Although relatively rare in absolute terms, this category remains particularly informative because it may capture gender-specific organizational responses following the emergence of \#MeToo. Changes in female executive turnover may reflect multiple channels including shifts in organizational culture, changing labor-market opportunities, or strategic restructuring efforts designed to signal governance responsiveness.
\\
\\
Finally, allegation-related disclosures constitute the least frequent category, averaging approximately 0.002 events per firm-month. Nevertheless, despite their rarity, such disclosures potentially convey substantial informational content. Because allegations may trigger reputational spillovers, litigation concerns, or heightened board scrutiny, even infrequent events can generate disproportionately large organizational responses. Taken together, the descriptive evidence suggests two important features of the data. First, executive turnover is characterized by substantial cross-firm heterogeneity rather than uniform organizational adjustment. Second, different dimensions of turnover exhibit markedly different frequencies and institutional meanings. These patterns provide preliminary motivation for examining whether the \#MeToo movement generated differential effects across organizational margins, governance structures, and executive categories rather than simply inducing an aggregate increase in turnover.
\\
\\
Figure 2 presents the distribution of firm-month observations across broad GICS sectors. Several features are noteworthy. First, the sample exhibits substantial sectoral diversification rather than concentration in a small number of industries. Industrial firms constitute the largest share of observations, accounting for approximately 16.4 percent of the sample, followed by information technology and financial firms, each representing roughly 13.6 percent. Health care and consumer discretionary firms also comprise substantial portions of the sample, accounting for approximately 12.8 and 10.4 percent of observations, respectively. The relatively balanced distribution across major sectors is important for several reasons. First, it mitigates concerns that the empirical findings are driven by a small number of highly visible industries or firms. Second, broad sectoral coverage strengthens external validity by ensuring that the analysis captures heterogeneous organizational environments characterized by different governance structures, labor-market conditions, and reputational exposure. Third, sectoral diversity becomes particularly relevant in the context of the \#MeToo shock because industries plausibly differ in both organizational culture and sensitivity to reputational pressure The distribution also suggests economically meaningful variation in potential treatment intensity. Sectors such as information technology, financial services, consumer-facing industries, and health care may operate under different governance environments and face different forms of public scrutiny compared with utilities or materials industries. Such variation provides useful motivation for the heterogeneous-treatment analysis presented subsequently.
\\
\\
Figure 3 reports the industry-level composition of the sample. The data reveal broad dispersion across industries, with no single industry accounting for a dominant share of observations. Capital markets firms represent the largest individual category at approximately four percent of the sample, followed by health-care equipment, multi-utilities, chemicals, and diversified financial industries. An important implication of this pattern is that the analysis is not driven by a narrow concentration of firms operating within a particular industry environment. Instead, the dataset spans industries characterized by substantially different technological structures, organizational hierarchies, labor-market arrangements, and governance systems. This heterogeneity is particularly valuable because the transmission of reputational-information shocks may vary systematically across industries. Industries with greater public visibility, stronger dependence on human capital, or greater reliance on reputation and consumer trust may respond differently to changes in social accountability norms than industries characterized by more standardized production structures or lower public exposure. The relatively diffuse industry composition therefore provides an empirical environment in which heterogeneous organizational responses can be identified while reducing concerns regarding industry-specific confounding factors.

\begin{figure}
    \centering
    \includegraphics[width=1.1\linewidth]{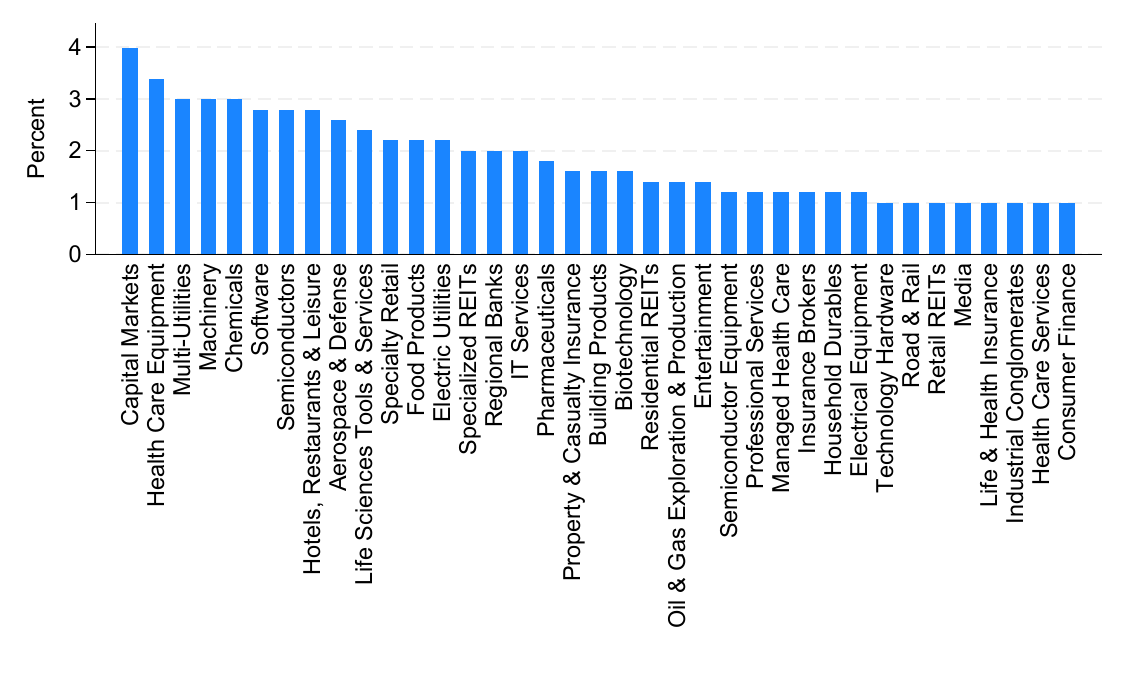}
    \caption{Industry Composition of the Executive Turnover Data}
    \label{fig:placeholder}
\end{figure}

Figure 4 presents the geographic distribution of firms across U.S. states. The sample exhibits substantial spatial variation, although several states account for a disproportionate share of observations. California represents the largest concentration, accounting for approximately 13.3 percent of the sample, followed by New York and Texas at approximately 10.2 and 9.8 percent, respectively. The concentration of firms within these states is unsurprising given their central roles within the U.S. corporate landscape. California hosts a substantial number of technology and innovation-intensive firms, New York remains a major center of financial and professional services activity, and Texas contains a large concentration of energy and industrial firms. At the same time, the sample extends across a broad range of states with differing legal environments, labor-market institutions, and industrial structures. Such geographic variation is particularly relevant because organizational responses to reputational shocks may interact with local institutional conditions, labor-market characteristics, and cultural norms. The geographic dispersion of firms therefore strengthens the external validity of the analysis and reduces concerns that estimated effects merely reflect localized shocks or region-specific organizational dynamics.
\\
\\
Taken together, the sectoral, industrial, and geographic distributions indicate that the sample spans a diverse set of organizational environments and institutional settings. The broad coverage reduces concerns that subsequent estimates reflect idiosyncratic features of a small number of sectors or geographic clusters and instead provides a rich setting for identifying heterogeneous organizational responses to reputational-information shocks.

\begin{figure}
    \centering
    \includegraphics[width=1.1\linewidth]{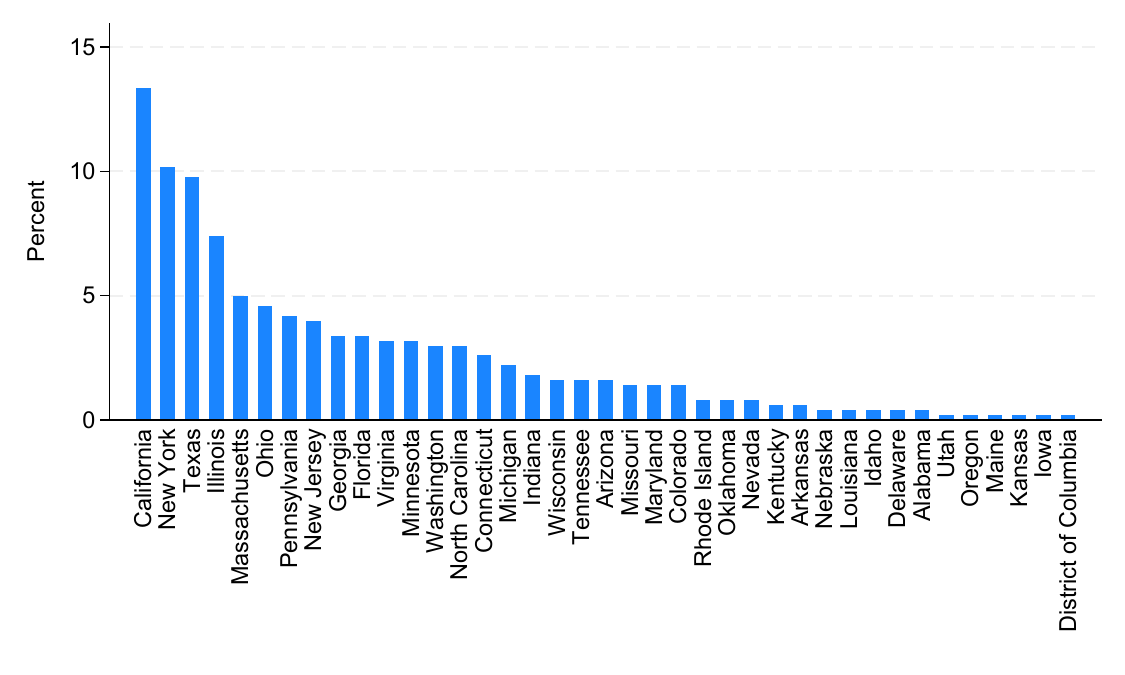}
    \caption{Geographic Distribution of Firm Headquarters Across U.S. States}
    \label{fig:placeholder}
\end{figure}

Figure 5 presents the evolution of resignation intensity for firms classified according to pre-treatment exposure status. Firms are categorized as exposed if they exhibited at least one relevant disclosure filing before October 2017 and as unexposed otherwise. The vertical dashed line denotes the Weinstein revelations in October 2017, which marked the emergence of the \#MeToo movement as a large-scale reputational-information shock.
\\
\\
Several features of the figure deserve attention. First, both exposed and unexposed firms exhibit relatively similar underlying trajectories during the pre-shock period. Prior to October 2017, resignation intensity evolves gradually and does not display evidence of abrupt divergence between the two groups. While exposed firms exhibit somewhat higher average resignation levels throughout the pre-treatment period, the overall trajectories remain broadly comparable and evolve smoothly over time. This pattern is important because it suggests that exposed firms were characterized by higher baseline organizational turnover rather than sharply divergent pre-treatment dynamics.
\\
\\
Second, following October 2017 the trajectories begin to separate more noticeably. Firms classified as exposed exhibit a substantially steeper increase in resignation intensity relative to unexposed firms. The divergence emerges gradually rather than instantaneously, suggesting that the organizational consequences of the \#MeToo shock were not limited to short-run reactions immediately following the Weinstein revelations. Instead, the effects appear to unfold progressively over time, consistent with organizational adaptation, internal investigations, governance restructuring, and reputational adjustment processes.
\\
\\
Third, the persistence of the divergence is economically informative. The gap between exposed and unexposed firms does not rapidly dissipate following the initial post-shock period. Rather, exposed firms continue to experience elevated resignation intensity throughout much of the subsequent period. Such persistence is difficult to reconcile with purely transitory media-attention effects and instead suggests more durable changes in organizational behavior and governance structures.
\\
\\
Importantly, the figure should not be interpreted as a formal test of the conventional parallel-trends assumption underlying binary-treatment difference-in-differences designs. Because the empirical framework exploits a common treatment timing with heterogeneous exposure intensity rather than staggered binary treatment adoption, standard pre-trend diagnostics are not directly applicable. Instead, the figure serves as a graphical diagnostic assessing whether firms characterized by different levels of pre-treatment exposure exhibit markedly different organizational trajectories prior to the shock itself. Overall, the evidence suggests that although exposed firms differ from unexposed firms in baseline resignation intensity, there is little indication of abrupt differential dynamics before October 2017. The more pronounced divergence emerging after the Weinstein revelations is consistent with the interpretation that the \#MeToo movement amplified pre-existing organizational vulnerabilities embedded within firms' disclosure environments.

\begin{figure}
    \centering
    \includegraphics[width=1.1\linewidth]{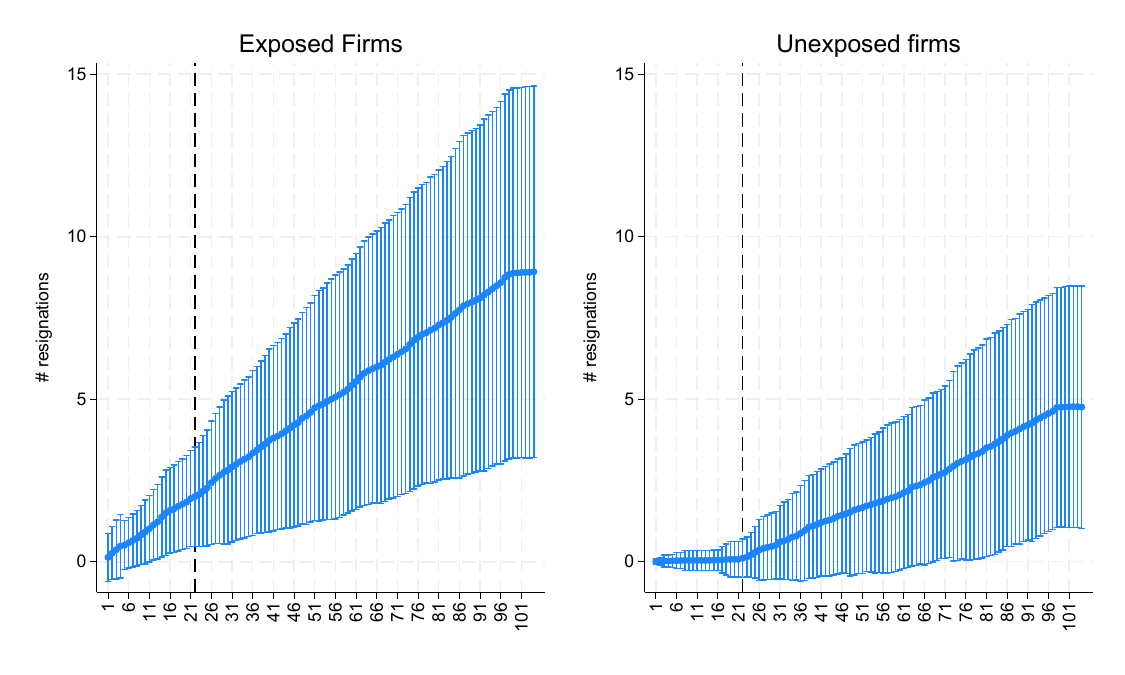}
    \caption{Pre- and Post-Shock Resignation Dynamics Across Exposure Groups}
    \label{fig:placeholder}
\end{figure}

\section{Results}

\subsection{Baseline difference-in-differences results}

Table 2 reports the baseline dynamic difference-in-differences estimates of the effect of the \#MeToo shock on corporate resignation intensity. The dependent variable is log resignation intensity, and all specifications include four lags of the outcome variable to account for persistence in firm-level turnover dynamics. The treatment variable is defined as the interaction between the post-Weinstein period and the predetermined exposure index constructed from pre-October-2017 SEC Form 8-K disclosures. Consequently, the coefficient identifies the differential post-shock response among firms with greater ex ante exposure, conditional on prior organizational turnover patterns.
\\
\\
Across all specifications, the estimated effects are positive and statistically significant, indicating that firms with greater pre-existing exposure experienced systematically larger increases in resignation intensity following the \#MeToo shock. Column (1) reports the baseline specification including firm fixed effects but excluding month effects. The estimated coefficient equals 0.300 and is statistically significant at the one percent level. Interpreted in percentage terms, this estimate implies approximately a 35 percent increase in resignation intensity among more exposed firms following the Weinstein revelations. This specification absorbs permanent differences across firms while allowing aggregate temporal variation to remain unrestricted.
\\
\\
Introducing month-fixed effects in column (2) substantially reduces the estimate to 0.035, corresponding to approximately a 3.6 percent increase in resignation intensity. The decline is economically informative because it suggests that a substantial fraction of the overall increase in executive turnover reflects common aggregate forces affecting all firms after October 2017. Nevertheless, even after absorbing these broad time effects, firms with greater ex ante exposure continue to exhibit significantly larger post-shock organizational responses.
\\
\\
Column (3) further incorporates industry fixed effects and produces an almost identical estimate of 0.035, again corresponding to approximately a 3.6 percent increase. The near invariance of the estimate suggests that the effect is not driven by systematic differences in industry composition. Additionally, column (4) introduces a substantially richer fixed-effect structure and produces an estimate of 0.123. This coefficient implies roughly a 13.1 percent increase in resignation intensity among exposed firms following the shock. The larger estimate suggests that once additional dimensions of unobserved heterogeneity are absorbed, the exposure-based effect becomes more pronounced.
\\
\\
Furthermore, column (5) introduces month-industry fixed effects and yields an estimated coefficient of 0.257, corresponding to approximately a 29.3 percent increase in resignation intensity. This specification effectively compares firms operating within similar sector-time environments and therefore substantially reduces concerns regarding industry-specific post-2017 shocks. Finally, the most saturated specification in column (6) produces a coefficient estimate of 0.113, implying approximately a 12 percent increase in resignation intensity following the \#MeToo shock. Importantly, this estimate remains highly statistically significant despite the demanding control structure and therefore represents a conservative estimate of the organizational effect.
\\
\\
Taken together, the estimates indicate economically meaningful effects rather than merely statistically detectable differences. Across specifications, the results imply increases in resignation intensity ranging from approximately 4 to 35 percent, with preferred saturated specifications suggesting effects in the range of 12 to 13 percent. These magnitudes are substantial given that executive turnover typically represents relatively infrequent and organizationally costly events. The evidence therefore suggests that the \#MeToo movement operated not merely as a broad social phenomenon but as a large-scale reputational-information shock that generated economically meaningful changes in organizational behavior among firms exhibiting greater ex ante governance and reputational vulnerability.

\begin{table}[!htbp]
\centering
\caption{Dynamic Difference-in-Differences Estimates of the \#MeToo Shock on Corporate Resignation Intensity, 01:2016-07:2024}
\label{tab:baseline_did}

\scriptsize
\setlength{\tabcolsep}{3.5pt}
\renewcommand{\arraystretch}{1.05}

\begin{threeparttable}

\resizebox{0.98\textwidth}{!}{
\begin{tabular}{lcccccc}
\toprule
& \multicolumn{6}{c}{Full Sample} \\
\cmidrule(lr){2-7}
& (1) & (2) & (3) & (4) & (5) & (6) \\
\midrule

Post-Weinstein $\times$ Exposure
& 0.300*** 
& 0.035** 
& 0.035** 
& 0.123*** 
& 0.257*** 
& 0.113*** \\

& (0.015)
& (0.013)
& (0.013)
& (0.005)
& (0.014)
& (0.000) \\

\midrule

Firm Fixed Effects                & YES & YES & YES & YES & NO  & YES \\
Month Fixed Effects              & NO  & YES & YES & YES & YES & YES \\
Industry Fixed Effects           & NO  & NO  & YES & NO  & YES & YES \\
Firm-Month Fixed Effects         & NO  & NO  & NO  & YES & NO  & NO  \\
Month-Industry Fixed Effects     & NO  & NO  & NO  & NO  & YES & NO  \\
Firm-Industry-Month Fixed Effects& NO  & NO  & NO  & NO  & NO  & YES \\

\midrule

\# Firms                & 498 & 498 & 498 & 498 & 498 & 498 \\
\# Months               & 91  & 91  & 91  & 91  & 91  & 91  \\
\# Paired Observations  & 45,317 & 45,317 & 45,317 & 45,317 & 45,317 & 45,317 \\

\bottomrule
\end{tabular}
}

\vspace{0.15cm}

\begin{minipage}{0.98\textwidth}
\scriptsize
\raggedright
\textit{Notes:} The table reports dynamic difference-in-differences estimates of the effect of the \#MeToo shock on firm-level resignation intensity. The dependent variable is the logarithm of resignation intensity. The treatment variable is defined as the interaction between the post-Weinstein period and a predetermined firm-level exposure index constructed from SEC Form 8-K filings submitted before October 2017. All specifications include four lags of the dependent variable to account for persistence in resignation dynamics. Columns differ in the set of fixed effects used to absorb unobserved heterogeneity across firms, months, industries, and higher-dimensional firm-industry-time structures. Cluster-robust standard errors are reported in parentheses. The sample consists of 498 firms observed over 91 months, yielding 45,317 firm-month observations. Asterisks denote significance levels are at 10\% (*), 5\% (**), and 1\% (***), respectively.
\end{minipage}

\vspace{0.15cm}
\hrule

\end{threeparttable}
\end{table}

Figure 6 presents the dynamic evolution of cohort-specific treatment effects following the Weinstein revelations in October 2017. The estimates are adjusted for firm and month fixed effects and therefore isolate differential post-shock organizational responses after accounting for persistent firm-specific characteristics and aggregate temporal variation. The figure reports average treatment effects over time together with associated confidence intervals.
\\
\\
Several features of the dynamic pattern deserve careful attention. First, the figure reveals an economically large and immediate organizational response following the emergence of the \#MeToo shock. Estimated treatment effects initially range between approximately 0.35 and 0.38 log points during the early post-treatment period, corresponding to roughly between 42 percent (=exp(0.35)-1)) and 46 percent (=exp(0.38)-1)) higher resignation intensity among exposed firms relative to less exposed firms. These magnitudes suggest a substantial initial restructuring response following the shock and imply effects that are difficult to interpret as merely short-lived fluctuations in managerial turnover.
\\
\\
Second, the effects remain highly persistent throughout much of the post-treatment period. Following the initial adjustment phase, estimated treatment effects stabilize at approximately 0.25--0.30 log points for an extended period between 2019 and 2021. These estimates imply persistent increases in resignation intensity of approximately 28 percent and 35 percent, relative to less exposed firms. The persistence of the effect is particularly important because it suggests that the organizational consequences of the \#MeToo movement extended well beyond immediate responses to media attention surrounding the Weinstein allegations. A purely transitory information shock would be expected to generate rapid reversion toward pre-treatment levels once public attention diminished. Instead, the estimates indicate prolonged organizational adjustment.
\\
\\
Third, beginning around 2021-2022 the magnitude of the treatment effect gradually declines. Importantly, this attenuation occurs smoothly rather than abruptly. Treatment effects progressively decrease toward approximately 0.02-0.05 log points by the end of the sample period, eventually becoming statistically indistinguishable from zero. The gradual convergence pattern provides important insight into the underlying adjustment mechanism. Rather than reflecting immediate one-time turnover events, the evidence is more consistent with an organizational adaptation process in which firms initially undertake substantial restructuring efforts following a large reputational-information shock and subsequently converge toward a new equilibrium.

\begin{figure}
    \centering
    \includegraphics[width=1\linewidth]{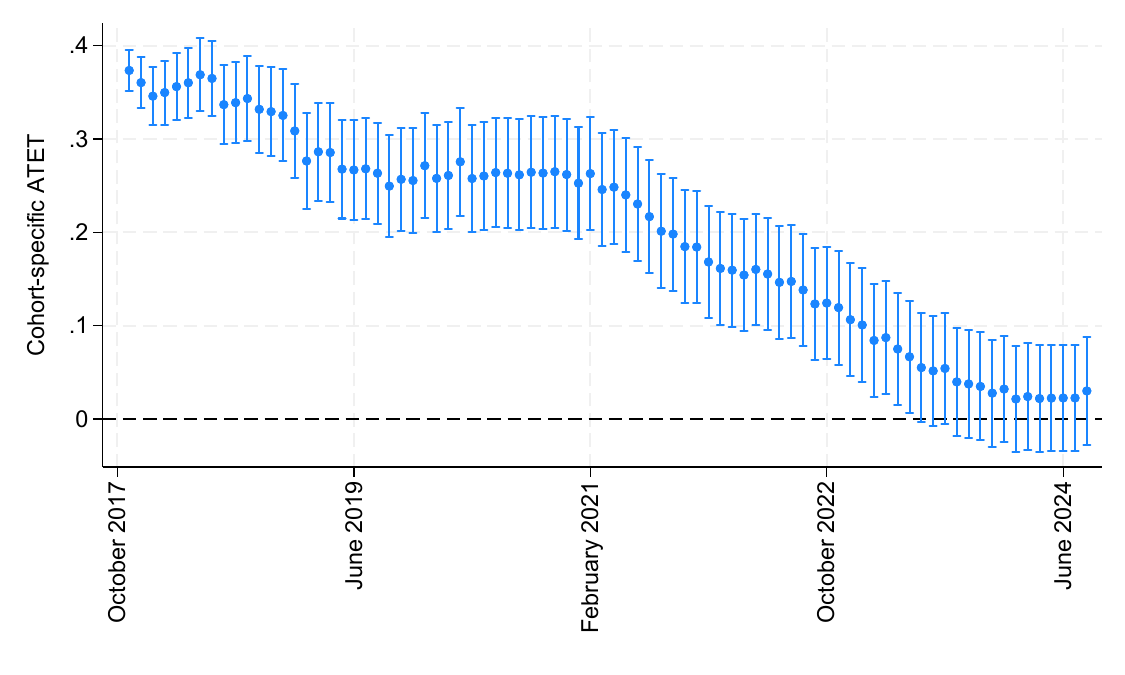}
    \caption{Dynamic Evolution of Cohort-Specific Treatment Effects Following the \#MeToo Shock}
    \label{fig:placeholder}
\end{figure}

Dynamic profile of post-exposure effect is difficult to reconcile with explanations based solely on contemporaneous media salience or temporary public attention. Instead, the evidence is more consistent with mechanisms involving (i) internal governance restructuring, (ii) organizational learning, (iii) compliance adaptation, (iv) revisions in monitoring practices, and (v) persistent changes in accountability norms. In this sense, the figure suggests that the \#MeToo movement operated not merely as a contemporaneous shock to public attention but rather as a catalyst for broader changes in organizational behavior and governance systems. Finally, the absence of abrupt oscillations or reversals in estimated effects provides additional reassurance regarding the stability of the empirical design. The treatment effect evolves smoothly over time and exhibits a coherent pattern of initial amplification followed by gradual attenuation, supporting the interpretation of a persistent but ultimately self-limiting organizational adjustment process.

\subsection{Two-way fixed-effect estimates}

Before turning to the baseline estimates, it is useful to consider the econometric challenges posed by the underlying panel structure. The empirical design combines repeated monthly observations for firms with persistent organizational outcomes and heterogeneous treatment intensity generated by pre-treatment exposure to the \#MeToo shock. Executive turnover decisions are unlikely to be independent across time because firms frequently experience episodes of sustained organizational restructuring in which one executive departure triggers subsequent managerial exits, board adjustments, or broader leadership realignment. Ignoring such persistence may therefore confound temporary fluctuations with systematic organizational adjustment processes. The inclusion of lagged outcomes follows the standard dynamic panel literature emphasizing the importance of accounting for persistence and state dependence in longitudinal settings (Arellano \& Bond 1991, Blundell \& Bond 1998).
\\
\\
To address these concerns, the analysis employs a dynamic two-way fixed effects framework combining firm and month fixed effects with lagged resignation intensity. Firm fixed effects absorb time-invariant organizational characteristics such as corporate culture, managerial structure, or governance quality, while month fixed effects account for aggregate shocks simultaneously affecting all firms, including macroeconomic conditions, labor-market developments, and broader changes in social accountability norms. The inclusion of lagged resignation intensity further ensures that identification derives from differential changes in resignation dynamics rather than permanent differences in baseline turnover levels across firms (Wooldridge 2010, Angrist \& Pischke 2009).
\\
\\
An additional concern relates to statistical inference. Conventional cluster-robust procedures may understate uncertainty when observations exhibit serial dependence within firms or when residual correlation persists across time. Because executive turnover dynamics may display substantial within-firm persistence and organizational shocks may generate correlated responses over multiple periods, asymptotic inference procedures may overstate statistical precision. To address this possibility, the analysis complements conventional two-way fixed effects estimation with wild-cluster bootstrap inference using increasing numbers of bootstrap replications. Wild-cluster bootstrap procedures have been shown to provide more reliable inference in panel environments characterized by clustered dependence structures and finite-sample distortions. Consequently, the bootstrap estimates reported below provide a more conservative assessment of statistical significance and allow evaluation of whether the estimated treatment effects remain robust under increasingly demanding inferential procedures.
\\
\\
Table 3 reports dynamic two-way fixed effects estimates of the effect of the \#MeToo shock on firm-level resignation intensity. The treatment variable is constructed as the interaction between the predetermined exposure index and the post-Weinstein period indicator. All specifications include firm and month fixed effects together with four lags of resignation intensity to account for persistence in organizational turnover dynamics.
\\
\\
Column (1) presents the conventional two-way fixed effects estimate using standard cluster-robust inference. The estimated coefficient equals 0.113 and is statistically significant at the one percent level. Because the dependent variable is specified in logarithmic form, the estimate implies approximately 12 percent (=exp(0.113)-1) higher resignation intensity among firms characterized by greater pre-treatment exposure following the emergence of the \#MeToo shock. Economically, this magnitude is substantial. Executive turnover events are relatively infrequent and often associated with considerable organizational costs arising from managerial replacement, restructuring, information loss, and governance adjustment. Consequently, an increase of approximately twelve percent in resignation intensity represents a nontrivial organizational response. Columns (2) through (5) examine the sensitivity of statistical inference using wild-cluster bootstrap procedures with increasing numbers of replications. The motivation for employing bootstrap inference arises from well-known concerns regarding finite-cluster distortions and dependence structures in panel estimators. Although conventional asymptotic standard errors often perform well in large samples, inference may become unreliable when residual dependence exists within firms over time.
\\
\\
Several findings deserve emphasis. First, the point estimate remains numerically identical across all specifications. Regardless of whether inference relies on conventional asymptotic approximations or bootstrap procedures, the estimated treatment effect remains fixed at 0.113. Such stability suggests that the estimated relationship reflects a substantive feature of the data rather than sensitivity to specification choices or inferential procedures. Second, the bootstrap t-statistics remain remarkably stable, varying only between approximately 2.45 and 2.48 despite increasing the number of replications from 100 to 1,000. This degree of stability is economically informative because it suggests that the statistical significance of the estimates is not driven by idiosyncratic sampling variation or small-cluster effects. Third, the results indicate that the positive effect of the \#MeToo shock survives under substantially more conservative inference procedures. The persistence of statistical significance under wild-cluster bootstrap inference strengthens confidence that the estimated treatment effect does not merely arise from residual within-firm correlation or finite-sample distortions.
\\
\\
Taken together, the evidence suggests that the estimated organizational response is not only economically meaningful but also statistically robust. Firms characterized by greater ex ante exposure experienced approximately twelve percent higher resignation intensity following the \#MeToo shock, and this conclusion remains essentially unchanged under increasingly demanding inference procedures. More broadly, the results reinforce the interpretation that the \#MeToo movement operated as a large-scale reputational-information shock whose organizational consequences propagated systematically through pre-existing differences in firm-level exposure.

\begin{table}[!htbp]
\centering
\caption{Dynamic Two-Way Fixed Effects Estimates with Wild-Cluster Bootstrap Inference}
\label{tab:twfe_wildbootstrap}

\scriptsize
\setlength{\tabcolsep}{4pt}
\renewcommand{\arraystretch}{1.05}

\begin{threeparttable}

\resizebox{0.98\textwidth}{!}{
\begin{tabular}{lccccc}
\toprule
& Standard TWFE & \multicolumn{4}{c}{Wild-Cluster Bootstrap (\# replications)} \\
\cmidrule(lr){2-2} \cmidrule(lr){3-6}
& (1) & 100 & 200 & 500 & 1,000 \\
&     & (2) & (3) & (4) & (5) \\
\midrule

Post-Weinstein $\times$ Exposure
& 0.113*** & 0.113*** & 0.113** & 0.113*** & 0.113*** \\
& (0.045) & (2.45) & (2.48) & (2.47) & (2.47) \\

\midrule

Firm Fixed Effects   & YES & YES & YES & YES & YES \\
Month Fixed Effects  & YES & YES & YES & YES & YES \\
Outcome Persistence  & YES & YES & YES & YES & YES \\
Persistence \(p\)-value & 0.000 & 0.000 & 0.000 & 0.000 & 0.000 \\

\bottomrule
\end{tabular}
}

\vspace{0.15cm}

\begin{minipage}{0.98\textwidth}
\scriptsize
\raggedright
\textit{Notes:} The table reports dynamic two-way fixed effects estimates of the impact of the \#MeToo shock on firm-level resignation intensity. The dependent variable is the logarithm of resignation intensity. The treatment variable is defined as the interaction between the post-Weinstein period indicator and a predetermined firm-level exposure index constructed exclusively from SEC Form 8-K filings submitted before October 2017. All specifications include four lags of the dependent variable to account for persistence in resignation dynamics, together with firm and month fixed effects. Column (1) reports conventional two-way fixed effects estimates with clustered standard errors reported in parentheses. Columns (2)-(5) report estimates based on wild-cluster bootstrap inference using increasing numbers of bootstrap replications. Parentheses in Columns (2)-(5) report bootstrap \(t\)-statistics. The sample consists of 498 firms observed monthly from January 2016 through July 2024. Asterisks denote significance levels are at 10\% (*), 5\% (**), and 1\% (***), respectively.
\end{minipage}

\vspace{0.15cm}
\hrule

\end{threeparttable}
\end{table}

\subsection{Mechanisms of Organizational Adjustment Following the \#MeToo Shock}

While the baseline estimates establish that firms with greater ex ante exposure experienced significantly higher post-shock resignation intensity, they do not by themselves identify the channels through which the \#MeToo shock propagated through organizations. Table 4 therefore examines the mechanisms underlying the observed increase in executive turnover by decomposing organizational responses into distinct dimensions of adjustment. Each specification estimates the effect of the \#MeToo shock on outcomes designed to capture specific organizational mechanisms and includes firm-industry fixed effects to absorb persistent heterogeneity across firms operating within similar industrial environments.
\\
\\
Several findings emerge. First, the aggregate effect reported in column (1) indicates that firms with greater exposure experienced approximately 12 percent (=exp(0.113)-1) higher resignation intensity following the shock. This estimate serves as the benchmark against which mechanism-specific effects can be interpreted and suggests that the \#MeToo movement generated economically meaningful organizational responses rather than isolated executive departures. Column (2) examines forced resignations as a proxy for board-enforced disciplinary governance. The estimated coefficient equals 0.175, implying approximately 19 percent (=exp(0.175)-1) higher forced resignation intensity among exposed firms. This effect represents the largest estimated mechanism in the table and provides particularly strong evidence for an accountability interpretation. Forced executive turnover is unlikely to arise from ordinary labor-market mobility and instead typically reflects active intervention by monitoring institutions. The evidence therefore suggests that boards and internal governance structures responded to the \#MeToo shock through more aggressive disciplinary actions and heightened scrutiny of executive conduct.
\\
\\
Column (3) examines voluntary resignations and captures anticipatory exits motivated by reputational self-insurance. The estimated coefficient of 0.127 implies approximately 13.5\% higher voluntary resignation intensity following the shock. This result suggests that organizational responses were not limited to externally imposed removals. Instead, executives may have internalized changes in reputational costs and adjusted behavior proactively by exiting organizations facing heightened scrutiny or uncertainty. Furthermore, column (4) examines board-member turnover as a mechanism of governance oversight restructuring. The estimated coefficient of 0.134 corresponds to approximately 14.3 percent higher board turnover among exposed firms. This finding is particularly informative because boards represent central monitoring institutions within firms. Increased turnover at the board level suggests that organizational responses extended beyond managerial replacement and involved broader restructuring of oversight mechanisms themselves.
\\
\\
In addition, column (5) examines female executive departures and identifies a mechanism related to organizational climate and labor-market sorting. The estimated coefficient of 0.133 implies approximately 14 percent higher turnover among female executives following the shock. Several interpretations are plausible. Female executives may have become more responsive to organizational culture and workplace conditions after \#MeToo, or firms may have altered leadership structures in response to changing expectations regarding diversity and governance accountability. Regardless of the specific channel, the results suggest that the organizational consequences of the shock differed across demographic dimensions of leadership.
\\
\\
Column (6) considers vice-president turnover as a measure of broader managerial restructuring. The estimated coefficient of 0.128 implies approximately 13 percent higher turnover intensity among vice presidents. This finding is important because vice presidents typically occupy intermediate positions within organizational hierarchies. Increased turnover among these executives suggests that the effects of \#MeToo propagated beyond highly visible senior leadership positions and affected broader managerial structures. In addition, column (7) examines chief financial officer departures and captures financial-governance and compliance responses. The estimated coefficient of 0.127 implies approximately 13 percent higher turnover among CFOs. Because CFOs occupy central positions in disclosure practices, compliance systems, and investor relations, the evidence suggests that the \#MeToo shock may also have generated organizational responses related to internal monitoring and risk management.
\\
\\
Finally, column (8) examines allegation-related disclosures as a proxy for informational amplification and reputational spillovers. The estimated coefficient of 0.149 corresponds to approximately 16 percent higher allegation-related disclosure activity among exposed firms. This result is particularly important because it suggests that \#MeToo may have affected not only realized organizational outcomes but also the informational environment surrounding firms themselves. Increased disclosure activity is consistent with a process through which reputational shocks amplified scrutiny and increased the probability that organizational vulnerabilities became publicly observable. Taken together, the estimates suggest that the \#MeToo movement did not generate a narrow response concentrated solely in executive turnover. Instead, the evidence points toward a broader process of organizational adjustment operating through multiple channels simultaneously, including disciplinary governance, anticipatory exits, restructuring of monitoring institutions, changes in managerial hierarchies, compliance responses, and informational amplification. The pattern of results therefore supports the interpretation that \#MeToo functioned as a large-scale reputational-information shock interacting with latent organizational vulnerabilities embedded within firms' pre-existing governance structures.

\begin{sidewaystable*}[p]
\centering
\thispagestyle{empty}

\caption{Mechanisms of Organizational Adjustment Following the \#MeToo Shock}
\label{tab:mechanisms}

\footnotesize
\renewcommand{\arraystretch}{1.12}
\setlength{\tabcolsep}{3.2pt}

\begin{threeparttable}

\begin{adjustbox}{max width=\textheight}
\begin{tabular}{lcccccccc}
\toprule
& All & Forced & Voluntary & Board & Female & Vice & Chief Financial & Flagged \\
& Resignations & Resignations & Resignations & Member & Executive & President & Officer & Allegation \\
& (1) & (2) & (3) & (4) & (5) & (6) & (7) & (8) \\
\midrule

Post-Weinstein $\times$ Exposure
& 0.113*** & 0.175*** & 0.127*** & 0.134*** & 0.133*** & 0.128*** & 0.127*** & 0.149*** \\
& (0.045) & (0.212) & (0.004) & (0.006) & (0.013) & (0.007) & (0.011) & (0.030) \\

\midrule
Firm-Industry Fixed Effects
& YES & YES & YES & YES & YES & YES & YES & YES \\

\midrule
\# Firms
& 498 & 498 & 498 & 498 & 498 & 498 & 498 & 498 \\

\# Months
& 99 & 99 & 99 & 99 & 99 & 99 & 99 & 99 \\

\# Paired Observations
& 45,317 & 45,317 & 45,317 & 45,317 & 45,317 & 45,317 & 45,317 & 45,317 \\

\midrule
Mechanism
& \parbox[c]{2.1cm}{\centering Governance\\ Destabilization}
& \parbox[c]{2.3cm}{\centering Board-Enforced\\ Disciplinary\\ Governance}
& \parbox[c]{2.5cm}{\centering Anticipatory Exit\\ \& Reputational\\ Self-Insurance}
& \parbox[c]{2.3cm}{\centering Governance\\ Oversight\\ Restructuring}
& \parbox[c]{2.3cm}{\centering Toxic Climate\\ \& Increased\\ Mobility}
& \parbox[c]{2.1cm}{\centering Mid-Level\\ Restructuring}
& \parbox[c]{2.5cm}{\centering Financial\\ Governance\\ \& Compliance\\ Response}
& \parbox[c]{2.5cm}{\centering Information\\ Revelation\\ \& Reputational\\ Amplification} \\

\bottomrule
\end{tabular}
\end{adjustbox}

\vspace{0.35cm}

\begin{minipage}{0.96\textheight}
\footnotesize
\raggedright
\textit{Notes:} The table reports dynamic difference-in-differences estimates of the mechanisms underlying the impact of the \#MeToo shock on firm-level organizational adjustment. The treatment variable is defined as the interaction between the post-Weinstein period indicator and a predetermined firm-level exposure index constructed exclusively from SEC Form 8-K filings submitted before October 2017. Outcome variables correspond to distinct dimensions of organizational response, including aggregate resignation intensity, forced resignations, voluntary resignations, board turnover, female executive turnover, vice-president turnover, chief financial officer turnover, and allegation-related disclosures. All specifications include four lags of the dependent variable to account for persistence in organizational turnover dynamics and firm-industry fixed effects to absorb persistent heterogeneity across firms operating within similar industrial environments. Cluster-robust standard errors at the firm level are reported in parentheses. The sample consists of 498 firms observed monthly from January 2016 through July 2024. Asterisks denote significance levels are at 10\% (*), 5\% (**), and 1\% (***), respectively.
\end{minipage}

\vspace{0.15cm}
\hrule

\end{threeparttable}
\end{sidewaystable*}

\subsection{Dynamic organizational adjustment and matrix completion estimation}

While the baseline difference-in-differences estimates provide evidence that firms with greater ex ante exposure experienced larger post-shock increases in resignation intensity, conventional fixed-effects approaches may impose restrictive assumptions regarding the evolution of counterfactual organizational trajectories. In particular, two-way fixed effects estimators assume that, conditional on observed controls and fixed effects, untreated outcomes evolve similarly across firms. Such assumptions may become restrictive in settings characterized by heterogeneous treatment intensity and dynamic organizational responses \citep{dechaisemartin2020twfe, sun2021estimating}.
\\
\\
These concerns are especially relevant in the present setting. The \#MeToo movement constituted a common shock occurring at a single point in time, while firms differed substantially in their pre-treatment exposure and organizational vulnerability. Consequently, treatment effects need not emerge uniformly across firms or sectors. Organizational responses may instead unfold gradually through internal investigations, changes in governance structures, managerial realignment, and revisions in monitoring practices. Such dynamics may generate complex counterfactual trajectories that are difficult to capture using standard additive fixed-effects structures. To accommodate these features, the analysis employs matrix completion methods to estimate dynamic treatment effects. Matrix completion estimators reconstruct missing counterfactual outcomes by exploiting latent low-dimensional structures embedded within panel data \citep{athey2021matrix}. Rather than assuming that untreated outcomes differ only through additive firm and time effects, the estimator allows counterfactual trajectories to be shaped by unobserved interactive components that vary simultaneously across firms and time. In practice, this approach permits substantially richer forms of heterogeneity and therefore provides greater flexibility for recovering organizational adjustment paths following the \#MeToo shock.
\\
\\
Conceptually, the matrix completion framework is particularly well suited for the present application because the treatment is characterized by common timing but heterogeneous exposure intensity. The estimator allows the data to recover the dynamic evolution of organizational responses without imposing restrictive assumptions regarding the functional form of treatment effects or the structure of latent organizational heterogeneity. Consequently, the estimates reported below provide insight not only into the magnitude of the treatment effect but also into the temporal evolution and persistence of organizational adjustment following the shock.
\\
\\
Table 5 reports dynamic matrix completion estimates of the effect of the \#MeToo shock on executive and board turnover. The table is informative because it moves beyond average treatment effects and asks a more demanding question regarding organizational adjustment: how large were the effects over time, and how did they vary across sectors? Unlike conventional fixed-effects approaches, matrix completion allows counterfactual organizational trajectories to be reconstructed using latent structures embedded within the panel and therefore accommodates richer forms of dynamic heterogeneity.
\\
\\
The full-sample estimate reported in Panel A is positive, statistically significant, and economically meaningful. The overall average treatment effect on the treated equals 0.201, implying approximately 22 percent (=exp(0.201)-1) higher resignation intensity among exposed firms following the \#MeToo shock. This estimate is larger than the preferred two-way fixed effects estimate, suggesting that once counterfactual organizational paths are reconstructed using latent interactive structures rather than additive fixed effects alone, the implied organizational response becomes stronger.
\\
\\
The sector-level estimates reveal substantial heterogeneity in the average response. Column (2) shows that Communication Services experienced an estimated increase of approximately 59 percent in resignation intensity. This sector is characterized by substantial public visibility and dependence on reputation-sensitive activities, making firms particularly vulnerable to changes in social accountability norms. Column (3) reports a more moderate but still economically meaningful effect of approximately 17 percent in Consumer Discretionary industries.
\\
\\
The largest average effects emerge in Consumer Staples in column (4), where the estimated coefficient of 0.704 implies approximately 102\% (=exp(0.704)-1) higher resignation intensity. Such a large estimate suggests that firms operating in sectors with substantial consumer interaction and reputation dependence may have faced stronger incentives to adjust organizational structures following the shock. Columns (5) through (9) reveal substantial variation across sectors. Energy and Communication Services display comparatively large responses of approximately 58 and 59 percent, respectively, while Financials, Industrials, and Information Technology exhibit positive but more moderate responses ranging between approximately 20 and 22 percent. By contrast, Health Care, Materials, and Real Estate show weaker or statistically insignificant average effects.
\\
\\
This pattern is informative because it suggests that \#MeToo did not generate a homogeneous corporate response. Rather, the shock propagated differentially across sectors depending on organizational structure, reputational exposure, workforce composition, and governance environments. Panel B provides particularly rich evidence regarding the temporal evolution of these responses. One month after the shock, the full-sample estimate equals 0.355, implying approximately 42.6 percent (=exp(0.355)-1) higher resignation intensity among exposed firms. The magnitude of this initial response is visible across nearly all sectors and is particularly pronounced in Utilities, Consumer Staples, Energy, Communication Services, Information Technology, and Real Estate.
\\
\\
Five months after the shock, the estimated full-sample effect remains virtually unchanged at approximately 40 percent, while ten months after the shock the estimate continues to imply approximately 39 to 40 percent higher resignation intensity. The persistence of the effect during the first year following the shock is important because it suggests that organizational responses extended beyond immediate reactions to the Weinstein revelations. The sector-specific dynamics reinforce this interpretation. Consumer Staples displays particularly large and persistent responses. The estimated effect increases from approximately 59 percent one month after the shock to more than 140 percent fifty months after the shock. Communication Services similarly exhibits large and persistent responses throughout the entire period. Energy also displays strong dynamic persistence, while Information Technology maintains large responses through approximately twenty months before attenuating toward the end of the sample period.
\\
\\
Fifteen and twenty months after the shock, the full-sample estimates gradually decline to approximately 37 and 30 percent, respectively. This attenuation appears economically plausible. Firms seem initially to react aggressively through executive replacement and organizational restructuring and subsequently move toward a more stable equilibrium. Perhaps the most important result emerges fifty months after the initial shock. The estimated full-sample effect remains positive and statistically significant at 0.183, implying approximately 20 percent (=exp(0.183)-1) higher resignation intensity relative to the counterfactual. The persistence of the effect nearly four years after the initial event is difficult to reconcile with interpretations based solely on temporary media attention or short-run reputational shocks.
\\
\\
Instead, the dynamic pattern appears more consistent with a process of prolonged organizational adaptation involving changes in monitoring structures, internal compliance systems, governance mechanisms, and managerial accountability. The gradual decline without abrupt reversal further suggests convergence toward a new organizational equilibrium rather than a temporary spike in executive turnover. Taken together, the matrix completion estimates indicate that the \#MeToo movement generated not merely an immediate increase in resignations, but rather a persistent process of organizational adjustment whose effects unfolded over multiple years and varied substantially across sectors. The results therefore support the broader interpretation of \#MeToo as a large-scale reputational-information shock interacting with latent organizational vulnerabilities embedded within firms' pre-existing governance environments.

\begin{sidewaystable*}[p]
\centering
\thispagestyle{empty}

\caption{Dynamic Matrix Completion Estimates of Organizational Adjustment Following the \#MeToo Shock}
\label{tab:matrix_completion}

\tiny
\renewcommand{\arraystretch}{1.08}
\setlength{\tabcolsep}{2.4pt}

\begin{threeparttable}

\begin{adjustbox}{max width=\textheight}
\begin{tabular}{lcccccccccccc}
\toprule
& All & Communication & Consumer & Consumer & Energy & Financials & Health & Industrials & IT & Materials & Real & Utilities \\
& Resignations & Services & Discretionary & Staples &  &  & Care &  &  &  & Estate &  \\
& (1) & (2) & (3) & (4) & (5) & (6) & (7) & (8) & (9) & (10) & (11) & (12) \\
\midrule

\multicolumn{13}{l}{\textit{Panel A: Full ATET}} \\

Overall
& 0.201*** & 0.465*** & 0.158 & 0.704*** & 0.456*** & 0.198** & 0.052 & 0.199 & 0.185* & -0.120 & -0.036 & 0.332 \\
& (0.048) & (0.162) & (0.152) & (0.205) & (0.183) & (0.114) & (0.138) & (0.162) & (0.107) & (0.213) & (0.164) & (0.245) \\

\midrule

\multicolumn{13}{l}{\textit{Panel B: Sector-Time-Specific ATET}} \\

One month after
& 0.355*** & 0.372*** & 0.385*** & 0.466*** & 0.407*** & 0.253*** & 0.246** & 0.321*** & 0.364*** & 0.343*** & 0.374*** & 0.477*** \\
& (0.027) & (0.071) & (0.043) & (0.053) & (0.063) & (0.089) & (0.097) & (0.063) & (0.068) & (0.052) & (0.083) & (0.049) \\

Five months after
& 0.336*** & 0.319*** & 0.354*** & 0.598*** & 0.445*** & 0.242*** & 0.334*** & 0.339*** & 0.368*** & 0.458*** & 0.221* & 0.329** \\
& (0.039) & (0.111) & (0.098) & (0.067) & (0.111) & (0.093) & (0.094) & (0.095) & (0.091) & (0.063) & (0.117) & (0.143) \\

Ten months after
& 0.334*** & 0.328*** & 0.423*** & 0.717*** & 0.523*** & 0.225*** & 0.235** & 0.201 & 0.411*** & 0.487*** & 0.182 & 0.322 \\
& (0.043) & (0.140) & (0.123) & (0.075) & (0.155) & (0.106) & (0.113) & (0.143) & (0.110) & (0.143) & (0.137) & (0.225) \\

Fifteen months after
& 0.317*** & 0.413*** & 0.226* & 0.815*** & 0.661*** & 0.196 & 0.157 & 0.252 & 0.491** & 0.450*** & 0.084 & 0.357 \\
& (0.051) & (0.163) & (0.140) & (0.078) & (0.159) & (0.119) & (0.144) & (0.158) & (0.129) & (0.170) & (0.159) & (0.265) \\

Twenty months after
& 0.263*** & 0.501*** & 0.123 & 0.606** & 0.597*** & 0.246** & 0.062 & 0.262 & 0.389** & 0.089 & 0.041 & 0.495* \\
& (0.052) & (0.171) & (0.144) & (0.241) & (0.175) & (0.124) & (0.153) & (0.171) & (0.166) & (0.218) & (0.177) & (0.267) \\

Fifty months after
& 0.183*** & 0.557*** & 0.112 & 0.885** & 0.550** & 0.187 & 0.051 & 0.187 & 0.077 & -0.331 & -0.109 & 0.317 \\
& (0.057) & (0.206) & (0.186) & (0.276) & (0.249) & (0.134) & (0.165) & (0.198) & (0.123) & (0.274) & (0.206) & (0.317) \\

\midrule

$\lambda$
& 0.0001 & 0.0001 & 0.0001 & 0.0001 & 0.0001 & 0.0001 & 0.0001 & 0.0001 & 0.0001 & 0.0001 & 0.0001 & 0.0001 \\

\# Firms
& 498 & 21 & 52 & 33 & 23 & 68 & 65 & 80 & 68 & 28 & 31 & 30 \\

\# Months
& 104 & 104 & 104 & 104 & 104 & 104 & 104 & 104 & 104 & 104 & 104 & 104 \\

\# Paired Observations
& 51,791 & 2,184 & 5,408 & 3,431 & 2,392 & 7,072 & 6,657 & 8,319 & 7,072 & 2,912 & 3,224 & 3,120 \\

Sample Fraction
& Full & 0.042 & 0.103 & 0.068 & 0.046 & 0.135 & 0.128 & 0.163 & 0.135 & 0.059 & 0.619 & 0.069 \\

\bottomrule
\end{tabular}
\end{adjustbox}

\vspace{0.3cm}

\begin{minipage}{0.96\textheight}
\tiny
\raggedright
\textit{Notes:} The table reports dynamic matrix completion estimates of the average treatment effect on the treated of the \#MeToo shock on firm-level resignation intensity. The dependent variable is the logarithm of resignation intensity. Treatment is defined by the interaction between the post-Weinstein period and a predetermined firm-level exposure index constructed from SEC Form 8-K filings submitted before October 2017. Panel A reports the full-period ATET for the full sample and separately by GICS sector. Panel B reports dynamic sector-time-specific ATETs one, five, ten, fifteen, twenty, and fifty months after the shock. Matrix completion reconstructs counterfactual resignation paths using untreated and lower-exposure observations while allowing for latent firm and time components. Standard errors are reported in parentheses. The penalty parameter is denoted by $\lambda$. Asterisks denote significance levels are at 10\% (*), 5\% (**), and 1\% (***), respectively.
\end{minipage}

\vspace{0.15cm}
\hrule

\end{threeparttable}
\end{sidewaystable*}

\begin{figure}
    \centering
    \includegraphics[width=1\linewidth]{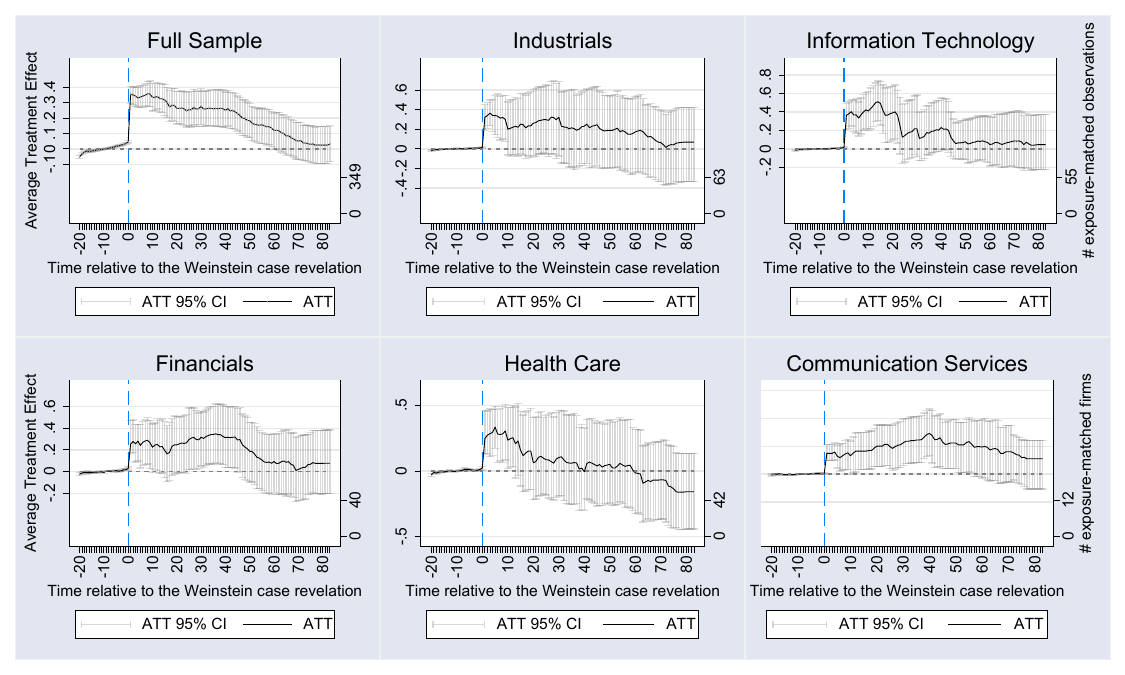}
    \caption{Dynamic Sectoral Treatment Effects of the \#MeToo Shock on Executive and Board Turnover}
    \label{fig:placeholder}
\end{figure}

The sectoral ATET plots in Figure 7 provide one of the clearest pieces of evidence that the \#MeToo shock did not operate as a uniform aggregate event. Instead, the shock generated a heterogeneous and temporally structured process of organizational adjustment. The full-sample panel shows a sharp increase in treatment effects immediately after the Weinstein revelations, followed by a long period of persistence and gradual attenuation. This pattern is important because it is inconsistent with a purely transitory media-attention shock. The response does not simply spike and disappear. It evolves as a multi-year adjustment process.
\\
\\
The full-sample estimate rises sharply around the event date and remains positive for much of the post-shock period. The profile suggests that exposed firms experienced an immediate increase in executive and board turnover, but the effect subsequently decayed only gradually. This dynamic shape is consistent with internal investigations, board review processes, reputational reassessment, compliance reforms, and sequential managerial replacement. In other words, the shock appears to have changed the organizational cost of retaining vulnerable executives rather than merely producing a short-lived burst of resignations. The sector-specific panels reveal substantial heterogeneity. Industrials display a positive and persistent post-shock response, but the effect is moderate and imprecisely estimated. This suggests that while industrial firms were affected, the adjustment process was more diffuse and less sharply concentrated than in reputation-intensive sectors. The broader confidence intervals are also consistent with heterogeneity across industrial subsectors, where firms differ substantially in workforce composition, public visibility, and governance structure.
\\
\\
Information Technology sector shows one of the most pronounced immediate responses. The ATET rises sharply after the event date and remains elevated during the early post-shock period before declining over time. This pattern is highly plausible. Technology firms tend to rely heavily on high-skill labor markets, public reputation, workplace culture, and talent retention. These features make them especially sensitive to reputational shocks involving workplace conduct, managerial behavior, and organizational norms. The attenuation later in the sample suggests that much of the adjustment occurred relatively early, consistent with rapid restructuring and reputational risk management. Financials also exhibit a clear positive response, with treatment effects rising after the shock and remaining elevated for an extended period before gradually declining. This pattern fits the sector's governance environment. Financial firms are heavily regulated, reputationally sensitive, and dependent on investor trust, compliance systems, and managerial credibility. The \#MeToo shock plausibly increased the cost of retaining executives associated with governance vulnerability or reputational risk, particularly in firms whose pre-treatment filings already signaled greater exposure.
\\
\\
In addition, health care sector shows a more muted and less persistent dynamic response. The treatment effect increases after the event but declines steadily and eventually approaches zero or becomes negative later in the sample. This pattern suggests that health-care firms may have responded initially but did not experience the same sustained organizational restructuring observed in technology, finance, or communication services. One interpretation is that health-care governance and compliance systems already operated under strong institutional and regulatory scrutiny, limiting the incremental effect of \#MeToo on executive turnover.
\\
\\
Most importantly, communication services presents perhaps the most theoretically compelling sectoral pattern. The post-shock response is positive, persistent, and relatively smooth. The effect rises gradually after the event, remains elevated for several years, and then attenuates only slowly. This is exactly the pattern one would expect in a sector where public visibility, media salience, brand reputation, and workplace culture are central organizational assets. In such sectors, reputational-information shocks can propagate through firms not merely by revealing individual misconduct, but by changing expectations about governance accountability and acceptable organizational conduct.
\\
\\
Taken together, the sectoral plots support three major findings. First, the \#MeToo shock produced an immediate organizational response among exposed firms. The break around October 2017 is visually apparent in the full sample and in several major sectors. Second, the effect was persistent. The post-shock dynamics do not resemble a short-lived media cycle. Instead, they suggest a sustained process of governance adjustment and executive turnover extending across multiple years. Third, the effect was heterogeneous across sectors. The strongest and most coherent patterns appear in sectors where reputation, public visibility, compliance, workplace culture, and human-capital retention are especially important. This sectoral heterogeneity strengthens the interpretation of \#MeToo as a reputational-information shock rather than a generic macroeconomic or labor-market disturbance. The plots therefore further reinforce the evidence for the paper's central mechanism: \#MeToo altered the organizational equilibrium by increasing the reputational and governance cost of retaining exposed executives, and this effect was strongest in sectors where public scrutiny and organizational culture are most consequential.

\subsection{Sectoral Sources of Organizational Adjustment}

While the preceding results establish that the \#MeToo shock generated economically meaningful increases in resignation intensity among exposed firms, an important remaining question concerns the extent to which the estimated effect is driven by particular sectors. Although previous analyses revealed substantial heterogeneity in sector-specific responses, a large average treatment effect could still potentially arise from a small number of sectors exhibiting unusually strong organizational reactions. To assess this possibility, Table 6 implements a leave-one-sector-out sensitivity analysis in which individual GICS sectors are sequentially excluded and the baseline dynamic two-way fixed effects specification is re-estimated. The objective of this exercise is not merely to assess robustness in a mechanical sense. Rather, it identifies the sectoral environments contributing most strongly to the estimated organizational response. If exclusion of a given sector substantially attenuates the estimated coefficient, that sector can be interpreted as contributing disproportionately to the observed effect. Conversely, if exclusion leaves the coefficient largely unchanged, the organizational response appears more broadly distributed throughout the economy.
\\
\\
Column (1) reports the baseline estimate of 0.114, implying approximately 12.1 percent higher resignation intensity following the \#MeToo shock among firms with greater pre-treatment exposure. The estimated coefficient remains positive and highly statistically significant across all exclusion exercises, immediately suggesting that the baseline result is not driven by any single sector. Nevertheless, the magnitude of attenuation differs meaningfully across specifications. The most pronounced reduction occurs when Communication Services are excluded. The estimated coefficient declines from 0.114 to 0.099, implying a reduction of approximately 13 percent relative to the baseline estimate. This attenuation is economically important because Communication Services firms are particularly exposed to public scrutiny, reputation-sensitive activities, media environments, and consumer-facing interactions. These characteristics make them natural candidates for stronger organizational responses following a shock involving workplace norms and accountability.
\\
\\
Substantial attenuation also occurs following exclusion of Energy and Consumer Staples firms. Excluding Energy reduces the coefficient to 0.094, while excluding Consumer Staples yields an estimate of 0.094 as well. These sectors therefore appear to contribute importantly to the aggregate organizational response. In the case of Energy, the result may reflect substantial organizational hierarchies and traditionally structured managerial systems where changes in accountability norms generate pronounced adjustment costs. Consumer Staples firms, by contrast, operate in highly visible consumer environments where reputational considerations may be especially important. In addition, Real Estate also exhibits meaningful influence, with the estimated coefficient increasing only modestly after exclusion. The reduction suggests that these firms contributed positively to the overall treatment effect, although to a lesser extent than Communication Services or Energy.
\\
\\
By contrast, excluding Industrials, Information Technology, Financial Services, Consumer Discretionary industries, or Utilities generates relatively small changes in coefficient magnitude. For example, exclusion of Industrials produces an estimate of 0.112 and exclusion of Information Technology yields 0.111, values almost identical to the baseline estimate. These findings suggest that although these sectors exhibit positive treatment effects individually, their contribution to the aggregate estimate is more diffuse and less concentrated. Perhaps the most important result is that no exclusion exercise eliminates the statistical significance of the estimated effect. Even under the most influential exclusion scenarios, the estimated coefficients remain positive, statistically significant, and economically meaningful. This finding substantially strengthens the credibility of the baseline results because it demonstrates that the estimated organizational response is not generated by a narrow subset of firms operating in unusually sensitive environments.
\\
\\
Taken together, the evidence suggests that the \#MeToo shock generated broad-based organizational responses across the corporate landscape, while certain sectors contributed disproportionately to the aggregate effect. The strongest contributors appear to be sectors characterized by high reputational exposure, intensive public scrutiny, consumer interaction, and organizational dependence on intangible assets. The resulting pattern is consistent with the broader interpretation of \#MeToo as a reputational-information shock interacting with latent organizational vulnerabilities rather than a localized phenomenon concentrated within a small set of industries.

\begin{sidewaystable*}[p]
\centering
\thispagestyle{empty}

\caption{Leave-One-Sector-Out Sensitivity Analysis of the \#MeToo Shock on Corporate Resignation Intensity}
\label{tab:leave_one_sector_out}

\normalsize
\renewcommand{\arraystretch}{1.45}
\setlength{\tabcolsep}{5pt}

\begin{threeparttable}

\begin{adjustbox}{width=0.98\textheight}
\begin{tabular}{lccccccccccc}
\toprule
& Baseline
& w/o Industrials
& w/o Information Technology
& w/o Financial Services
& w/o Health Care
& w/o Consumer Disc.
& w/o Consumer Staples
& w/o Real Estate
& w/o Utilities
& w/o Energy
& w/o Communication Services \\
& (1) & (2) & (3) & (4) & (5) & (6) & (7) & (8) & (9) & (10) & (11) \\
\midrule

Post-Weinstein $\times$ Exposure
& 0.114**
& 0.112**
& 0.111**
& 0.112***
& 0.137***
& 0.117**
& 0.094**
& 0.129
& 0.108***
& 0.094***
& 0.098*** \\

& (0.045)
& (0.048)
& (0.049)
& (0.049)
& (0.049)
& (0.048)
& (0.047)
& (0.047)
& (0.047)
& (0.047)
& (0.047) \\

\midrule

Firm Fixed Effects
& YES & YES & YES & YES & YES & YES & YES & YES & YES & YES & YES \\
\(p\)-value
& 0.000 & 0.000 & 0.000 & 0.000 & 0.000 & 0.000 & 0.000 & 0.000 & 0.000 & 0.000 & 0.000 \\

\midrule

Time Fixed Effects
& YES & YES & YES & YES & YES & YES & YES & YES & YES & YES & YES \\
\(p\)-value
& 0.000 & 0.000 & 0.000 & 0.000 & 0.000 & 0.000 & 0.000 & 0.000 & 0.000 & 0.000 & 0.000 \\

\midrule

Outcome Persistence
& 0.000 & 0.000 & 0.000 & 0.000 & 0.000 & 0.000 & 0.000 & 0.000 & 0.000 & 0.000 & 0.000 \\

\midrule

\# Firms
& 498 & 418 & 430 & 430 & 434 & 446 & 465 & 467 & 468 & 475 & 477 \\

\# Months
& 104 & 104 & 104 & 104 & 104 & 104 & 104 & 104 & 104 & 104 & 104 \\

\# Paired Observations
& 45,317 & 38,037 & 39,129 & 39,129 & 39,493 & 40,585 & 42,315 & 42,496 & 42,587 & 43,224 & 43,406 \\

\bottomrule
\end{tabular}
\end{adjustbox}

\vspace{0.7cm}

\begin{minipage}{0.93\textheight}
\scriptsize
\raggedright
\textit{Notes:} The table reports leave-one-sector-out sensitivity estimates of the impact of the \#MeToo shock on firm-level resignation intensity. Each specification sequentially excludes one GICS sector and re-estimates the baseline dynamic two-way fixed effects model. The dependent variable is the logarithm of resignation intensity. The treatment variable is defined as the interaction between the post-Weinstein period indicator and a predetermined firm-level exposure index constructed from SEC Form 8-K filings submitted before October 2017. All specifications include four lags of the dependent variable, firm fixed effects, and month fixed effects. Standard errors clustered at the firm level are reported in parentheses. Outcome-persistence \(p\)-values correspond to tests of lagged resignation dynamics. Asterisks denote significance levels are at 10\% (*), 5\% (**), and 1\% (***), respectively..
\end{minipage}

\vspace{0.15cm}
\hrule

\end{threeparttable}
\end{sidewaystable*}

\subsection{Discussion}

The empirical evidence presented above suggests that the \#MeToo movement generated organizational consequences extending substantially beyond isolated episodes of executive turnover. Across a broad range of specifications, firms characterized by greater ex ante exposure experienced significantly larger post-shock increases in resignation intensity, with effects persisting for multiple years and propagating through several organizational dimensions, including disciplinary turnover, board restructuring, managerial reorganization, compliance responses, and informational amplification. The dynamic and heterogeneous nature of these responses suggests that the \#MeToo movement functioned not merely as a contemporaneous media event, but rather as a large-scale reputational-information shock interacting with latent organizational vulnerabilities embedded within firms' governance structures.
\\
\\
The findings contribute directly to the literature on executive turnover and corporate governance. Traditional agency-based frameworks emphasize the role of boards and monitoring institutions in disciplining managerial behavior and reducing agency conflicts (\citep{jensen1976firm, fama1983ownership}). Within this literature, executive departures are often interpreted as responses to poor firm performance, governance failures, or changing organizational incentives (\citep{warner1988management, weisbach1988outside}). The present evidence suggests an additional mechanism through which turnover decisions emerge. Organizational accountability appears responsive not only to internal firm characteristics but also to broader shifts in social norms and reputational environments. Firms with greater ex ante exposure did not merely experience increased turnover rates but exhibited systematic restructuring across multiple layers of organizational hierarchy, suggesting that changes in social accountability norms can alter the incentives governing executive retention and monitoring decisions.
\\
\\
The results also relate to a broader literature examining reputational penalties and informational shocks in corporate environments. Prior work demonstrates that disclosure of misconduct can generate substantial economic costs through reputational damage, litigation exposure, and investor responses (\citep{karpoff2008cost}). Related research shows that firms adjust behavior in response to changes in informational environments and uncertainty (\citep{baker2016uncertainty, hassan2019political}). The evidence presented here complements these findings by demonstrating that informational shocks may also reshape internal organizational structures themselves. The observed increases in executive turnover and board restructuring suggest that reputational shocks propagate not only through market outcomes but also through managerial and governance channels.
\\
\\
An important contribution of the paper concerns the dynamic nature of organizational adjustment. Existing work examining media attention or public scandals frequently emphasizes short-run reactions and immediate responses following salient events (\citep{tetlock2007sentiment, tetlock2008fundamentals}). The dynamic matrix completion estimates presented above indicate a substantially different pattern. Organizational responses emerge rapidly but remain persistent over extended periods before gradually attenuating toward a new equilibrium. Such patterns are difficult to reconcile with purely transitory media-attention mechanisms. Instead, they appear more consistent with organizational learning and institutional adaptation processes in which firms revise monitoring structures, governance arrangements, and managerial composition over time.
\\
\\
The evidence further contributes to a growing literature using high-dimensional textual information to recover latent institutional characteristics from disclosure environments (\citep{gentzkow2019text, hoberg2016text, hassan2019political}). Rather than relying on ex post scandal identification or realized misconduct, the empirical design developed here constructs a predetermined exposure measure using computational analysis of pre-treatment SEC filings. The resulting framework suggests that mandatory disclosure text contains economically meaningful information regarding latent organizational vulnerability that becomes consequential when firms are exposed to large informational shocks. In this sense, the paper contributes not only to the study of organizational adjustment but also to the broader integration of computational methods into empirical corporate governance research.
\\
\\
More broadly, the findings speak to a wider literature examining how social norms and institutional environments shape economic behavior. Changes in social expectations can alter the implicit constraints governing organizational conduct and managerial incentives even in the absence of formal regulatory intervention. The \#MeToo movement appears to have changed the expected costs associated with maintaining organizational structures perceived as vulnerable to reputational scrutiny. The resulting effects therefore extend beyond the immediate consequences of individual allegations and instead reflect broader shifts in accountability norms within corporate environments. Taken together, the evidence suggests that the \#MeToo movement should not be viewed solely as an isolated social event affecting individual firms or executives. Rather, it appears to have acted as a large-scale informational and institutional shock that generated persistent organizational responses among firms with greater pre-existing vulnerabilities. The results imply that reputational shocks may alter internal organizational structures in ways that extend far beyond immediate market reactions and can produce durable changes in managerial behavior, governance systems, and accountability mechanisms.

\section{Conclusion}

This paper began with a simple but difficult question. Can a broad social movement alter corporate governance in a systematic and persistent manner when the underlying sources of misconduct are often hidden from direct observation? The challenge is particularly acute in corporate environments where disclosure rules frequently reveal that executives leave but provide little information regarding why they leave. Under such conditions, identifying the consequences of a common reputational shock requires moving beyond traditional treatment-control comparisons and toward a framework capable of exploiting heterogeneous exposure to a shared informational event.
\\
\\
We conceptualize the October 2017 Weinstein revelations and the emergence of the \#MeToo movement as an economy-wide reputational-information shock that altered beliefs regarding the expected costs of misconduct and intensified scrutiny within firms. Because all firms were simultaneously exposed to this shift in the informational environment, identification relied on differences in firms' ex ante vulnerability to governance-related disclosure and reputational pressure. We developed a theoretical framework in which board members and executives respond to changes in reputational incentives through forward-looking hazard rates, Bayesian belief updating, and strategic complementarities among organizational actors. The resulting model predicts that common informational shocks need not generate uniform responses. Instead, organizational reactions should depend on pre-existing exposure and internal vulnerability.
\\
\\
The empirical evidence strongly supports this prediction. Firms exhibiting greater pre-treatment exposure experienced systematically larger increases in executive and board turnover following the \#MeToo shock. These effects emerge rapidly, persist over multiple years, and propagate through multiple organizational dimensions, including forced departures, board restructuring, compliance responses, managerial reorganization, and informational amplification. Dynamic estimates indicate that the effects gradually attenuate over time but do not disappear, suggesting a process of sustained organizational adjustment rather than a temporary reaction to media attention.
\\
\\
Several broader implications follow. First, the findings suggest that reputational shocks operate through mechanisms extending far beyond individual misconduct. The \#MeToo movement did not simply expose particular executives or firms. It altered the expected consequences of remaining associated with organizational environments perceived as vulnerable to scrutiny. As those expectations changed, boards, managers, and firms adapted.
\\
\\
Second, the evidence suggests that informational environments themselves constitute an important source of institutional heterogeneity. Two firms facing the same external event need not respond similarly if they differ in their underlying governance structures, disclosure histories, or organizational vulnerability. The consequences of large social shocks therefore depend not only on the shock itself, but also on the latent institutional structures through which the shock propagates.
\\
\\
Third, the paper illustrates the value of integrating computational methods into causal analysis of institutions and organizations. By transforming large-scale disclosure text into a predetermined measure of organizational exposure, the analysis demonstrates how high-dimensional information can be used to identify otherwise unobservable dimensions of governance and reputational risk. As textual and digital data continue to expand, such approaches may substantially reshape the study of institutional change and organizational behavior.
\\
\\
More broadly, the results speak to an increasingly important feature of modern economies. Institutions are shaped not only by laws, regulations, and formal interventions. They are also shaped by information, expectations, and social norms. Major informational events can alter incentives even when formal rules remain unchanged. In such environments, organizations do not simply react to observable misconduct. They react to changing beliefs about what society is willing to tolerate. The broader lesson is therefore straightforward. The most consequential institutional changes are not always written into statutes or imposed by regulators. Sometimes they emerge from shifts in collective beliefs that alter the incentives governing behavior itself. When that occurs, information becomes governance, and reputational pressure becomes a mechanism of institutional change.

\bibliography{references}

\end{document}